\documentclass[10pt]{amsart}
\usepackage{amssymb}
\usepackage{amsmath}
\usepackage{verbatim}
\usepackage{enumerate}
\usepackage{amscd}

\newcommand{\dd}{\partial}
\newcommand{\be}{\begin{equation}}
\newcommand{\ee}{\end{equation}}
\newcommand{\ra}{\rightarrow}

\newcommand{\RR}{\mathbb{R}}
\newcommand{\ZZ}{\mathbb{Z}}
\newcommand{\EE}{\mathcal{E}}
\newcommand{\DD}{\mathcal{D}}
\newcommand{\CC}{\mathcal{C}}
\newcommand{\n}{m}
\newcommand{\E}{E_2}
\newcommand{\mr}{\mathrm}
\newcommand{\bt}{\bullet}
\newcommand{\hra}{\hookrightarrow}
\newcommand{\glob}{\mr{glob}}
\newcommand{\In}{\mr{in}}
\newcommand{\Out}{\mr{out}}
\newcommand{\uF}{\underline{F}}
\newcommand{\uC}{\underline{C}}
\newtheorem{thm}{Theorem}[section]
\newtheorem*{thm*}{Theorem}
\newtheorem{lemma}[thm]{Lemma}
\newtheorem{Conj}[thm]{Conjecture}
\newtheorem{remark}[thm]{Remark}
\newtheorem{proposition}[thm]{Proposition}

\begin{document}
\title
{
Wave relations
}

\author[A.~S.~Cattaneo]{Alberto~S.~Cattaneo}
\address{Institut f\"ur Mathematik, Universit\"at Z\"urich,
Winterthurerstrasse 190, CH-8057 Z\"urich, Switzerland}
\email{alberto.cattaneo@math.uzh.ch}

\author[P. Mnev]{Pavel Mnev}
\address{Institut f\"ur Mathematik, Universit\"at Z\"urich,
Winterthurerstrasse 190, CH-8057 Z\"urich, Switzerland}
\email{pmnev@pdmi.ras.ru}

\thanks{A.~S.~C. acknowledges partial support of SNF Grant No.~200020-131813/1.
P.~M. acknowledges partial support of RFBR Grant
No.
13-01-12405-ofi-m-2013
and of SNF Grant No.~200021-137595.}

\begin{abstract}
The wave equation (free boson) problem 
is studied from the viewpoint of the relations on the symplectic manifolds associated to the boundary
induced by solutions. Unexpectedly there is still something to say on this simple, well-studied problem.
In particular, boundaries which do not allow for a meaningful Hamiltonian evolution are not problematic from the viewpoint of relations.
In the two-dimensional Minkowski case, these relations are shown to be Lagrangian. This result is then extended to a wide class of metrics and is conjectured to be true also
in higher dimensions for nice enough metrics. A counterexample where the relation is not Lagrangian is provided by the Misner space.
%
%
\end{abstract}

\maketitle

\tableofcontents

\section{Introduction}
In this note we study the wave equation from the point of view of evolution relations (as defined in \cite{CMR, CMR2}). In particular we show that
they are well behaved also in cases when the boundary does not allow for a meaningful Hamiltonian evolution. This is a case study for a simple
well-studied problem (on which unexpectedly there was still something to say) supporting the relevance of the evolution relation approach.
This note is self-contained and the relevant concepts from \cite{CMR, CMR2} are introduced when needed.

Fix a dimension $\n$ and a signature. To an $\n$-dimensional 
compact oriented pseudo-Riemannian manifold $(M,g)$, possibly with boundary, whose metric has the given signature, we associate a space of fields\footnote{We mainly consider smooth functions in this note, which requires working in the setting of Fr\'echet spaces. For less regularity, and the
corresponding Banach setting, see subsubsection~\ref{ss:BvF}.
}
$F_M:=\CC^\infty(M)$ and an action functional
\be \label{S}
S_{M,g}[\phi]:=\frac12\int_M  d\phi\wedge*_g d\phi,\qquad\phi\in F_M,
\ee
where $*_g$ denotes the Hodge-$*$ operator induced by the metric $g$. More explicitly, writing the integrand in a local chart,
\[
S_{M,g}[\phi]=\frac12\int_M g^{\mu\nu}\,\dd_\mu\phi\,\dd_\nu\phi\,\sqrt g\; d^\n x.
\]
where $g=|\det (g_{\mu\nu})|$.
According to the construction in \cite{CMR2}, to an $(\n-1)$-manifold $\Sigma$ (with the extra structure of a function, a vector field and a volume form) we can associate a space of boundary fields\footnote{Notice that $F_M$ and $\Phi_\Sigma$ are Fr\'echet spaces and hence Fr\'echet manifolds.} (or phase space) $\Phi_\Sigma$ endowed with
a symplectic structure $\omega_\Sigma=\delta\alpha_\Sigma$, where $\alpha_\Sigma$ is the 1-form on $\Phi_\Sigma$ arising as the boundary term of the variation of the action (\ref{S}), such that for every $M$ as above we get an epimorphism (and hence a surjective submersion)
$\pi_M: F_M\to \Phi_{\dd M}$; moreover, $L_{M}:=\pi_M(EL_M)$ is isotropic in
$(\Phi_{\dd M},\omega_{\dd M})$,
where $EL_M$ is the subset of solutions to the  Euler--Lagrange (EL) equation $d*_g d \phi=0$ or in local coordinates:
\be\label{wave eq covariant}
\dd_\mu(\sqrt g\, g^{\mu\nu}\dd_\nu\phi)=0.
\ee
\begin{Conj}\label{Conj}
For any compact oriented pseudo-Riemannian manifold $(M,g)$ with boundary, which can be isometrically embedded into some Euclidean space $\RR^N$ equipped with constant metric, the subspace $L_M\subset \Phi_{\dd M}$ is Lagrangian.
\end{Conj}
The Conjecture is easily proved in the case of Riemannian manifolds \cite{CMR} from existence and uniqueness of solutions to the Dirichlet boundary problem $d*d \phi =0$, $\phi|_{\dd M}=\phi_\dd$ where $\phi_{\dd}\in \CC^\infty(\dd M)$ is a boundary condition for $\phi$. 

The Conjecture is also true if $M$ is a pseudo-Riemannian manifold of the form $\Sigma\times I$, where $I=[t_0,t_1]$ is an interval,
provided that the metric $g$ pulls back to a non-degenerate metric $g_t$ on $\Sigma\times \{t\}$ for any $t\in I$.
This follows from
existence and uniqueness for the initial value problem for the EL equation, with Cauchy data being a point of $\Phi_{\Sigma\times\{t_0\}}$, implying that $L_M$ is the graph of an isomorphism $\Phi_{\Sigma\times\{t_0\}}\ra\Phi_{\Sigma\times\{t_1\}}$. On the other hand, $L_M$ is isotropic by a universal argument of \cite{CMR, CMR2}, hence $L_M$ is Lagrangian.

If the boundary of $M$ is split into ``incoming'' and ``outgoing'' parts $\dd M=(\dd_\In M)^\mr{op}\sqcup \dd_\Out M$, then the subspace $L_M\subset \Phi_{\dd M}=\bar\Phi_{\dd_\In M}\times \Phi_{\dd_\Out M}$ can be interpreted as a set-theoretic relation between $\Phi_{\dd_\In M}$ and $\Phi_{\dd_\Out M}$ -- the ``evolution relation''; the property that $L_M$ is Lagrangian means that the evolution relation is a canonical relation between symplectic spaces. Here ``op'' stands for reversing the orientation and bar stands for changing the sign of the symplectic form. Gluing of manifolds along common boundary corresponds in this setting to set-theoretical composition of relations.

In this note, see Section~\ref{sec: compact}, Theorem \ref{thm}, we prove
the following case of the Conjecture.
\begin{thm*}\label{t:t}
Let $\n=2$, 
$M$ a compact domain with smooth boundary in the Minkowski plane, such that there are only finitely many boundary points with light-like tangent and such that the curvature of the boundary is nonzero at these points, and let $g$ be the Minkowski metric restricted to $M$. Then $L_M$ is Lagrangian.
\end{thm*}


Notice that $\dd M$ generally has several space-like and several time-like pieces separated by 
light-like points.
A consequence of the theorem is that we can study the wave equation on compact domains: the appropriate boundary conditions consist in the
choice of an affine Lagrangian subspace $L'$ of $\Phi_{\dd M}$ that intersects $L_M$ in one point and on which
$\alpha_{\dd M}+\delta f$ vanishes for some local functional $f$ on $\Phi_{\dd M}$. This also means that quantization on compact domains
is possible, provided a suitable polarization of $\Phi_{\dd M}$ can be found.

It is tempting to extend the conjecture also to more general pseudo-Riemannian manifolds. This is expected to be the case for metrics that are nice enough, e.g., which do not differ much from the constant one (which is the case for small domains). However, this is not true in general and in subsection~\ref{ss-Misner} we show that the Misner space \cite{Mis} provides
a counterexample. From a different perspective, we may say that the condition that $L$ is Lagrangian selects reasonable spacetimes.

If $M$ is of the form $\Sigma\times [t_0,t_1]$, one can attempt to define the Hamiltonian evolution in the ``time'' parameter $t\in [t_0,t_1]$. The Hamiltonian $H$ can be constructed in a standard way via the Legendre transform of the time-density of the action. Generally, for $M$ pseudo-Riemannian and with no non-degeneracy condition on the metric pulled back to time-slices $\Sigma\times\{t\}$, $H$ will be singular and one can employ the Gotay-Nester-Hinds algorithm \cite{GNH} to construct a smaller phase space on which $H$ and the associated Hamiltonian vector field are well-defined. However, generally the Hamiltonian vector field cannot be integrated to a finite-time Hamiltonian flow: both existence and uniqueness of solutions of the Hamilton's equations can fail. In Section \ref{sec: Hamiltonian for a circle} we study in detail an example of this situation: the radial evolution on an annulus on Minkowski plane. Despite the failure of the Hamiltonian picture in such cases, the formalism of canonical relations works perfectly and provides a more general framework for describing the evolution.

\subsection{Plan of the paper}
\begin{itemize}
\item In Section~\ref{sec: free boson} we review the construction of the boundary phase space for the classical field theory defined by action (\ref{S}) on a general pseudo-Riemannian manifold with boundary.
\item In Section \ref{sec: Minkowski plane} we specialize to the case of a domain $M$ in the Minkowski plane and study several simple examples explicitly. In particular,
we show that $L_M$ is Lagrangian if $M$ is a strip.\footnote{Since this is a noncompact manifold, appropriate restrictions on the behavior of fields at infinity are to be imposed.} In the case when the boundary is light-like, we observe however that no choice of boundary condition leads to uniqueness of solutions.
We also consider a diamond on the Minkowski plane with edges aligned in light-like directions and show that $L_M$ for this domain is Lagrangian.
\item 
Section \ref{sec: compact}
is central to this paper.
Here we specialize further to the case of compact domains on the Minkowski plane bounded by a collection of smooth curves with only finitely many light-like points (with the technical requirement that the boundary should have non-zero curvature at the light-like points). We prove that $L_M$ is Lagrangian for such domains (Theorem \ref{thm}).
\item In Section \ref{sec: remarks} we comment on several associated issues, in particular:
\begin{itemize}
\item Problems with Dirichlet boundary conditions (non-transversality of the corresponding $L'$ and $L_M$).
\item Constraint (Cauchy) subspaces of the phase space (constraints arising from the requirement of extendability of boundary fields to a solution of the wave equation in an open neighborhood of the boundary).
\item Conformal invariance of the problem. In particular, the result of Theorem \ref{thm} extends to domains with a non-flat Lorentzian metric conformally equivalent to the flat one.
\item The Hamiltonian formalism corresponding to radial evolution on the plane and issues with integrating the corresponding Hamiltonian vector field into a flow (both in the Fr\'echet and in the Banach setting).
\item The representation of the operad of little 2-disks by canonical relations coming from evolution relations $L_M$.
\item Interpretation of the property of being Lagrangian for the evolution relation for a general classical free field theory, possibly with gauge symmetry, in terms of (generalized) Lefschetz duality, and the specialization to the theory defined by the action (\ref{S}).
\item Extension of the result of Theorem \ref{thm} to more general Lorentzian surfaces, satisfying certain constraints on the metric.
\item An example of a Lorentzian surface with a non-Lagrangian evolution relation -- the Misner metric on a cylinder.
\end{itemize}
\end{itemize}

\subsection{Acknowledgements}
We thank C. De Lellis, T. Kappeler, V. Schroeder and  A. Weinstein for useful discussions.

\section{Classical massless free boson on a pseudo-Riemannian manifold: boundary structures} \label{sec: free boson}

The construction of \cite{CMR2} in case of the free massless boson, see action (\ref{S}), on a pseudo-Riemannian manifold
associates to a closed oriented $(\n-1)$-manifold $\Sigma$ endowed with a triple of a function, a vector field and a volume form $(\Gamma,u,\mu)\in \CC^\infty(\Sigma)\times \mathfrak{X}(\Sigma)\times \Omega^{\n-1}(\Sigma)$, a \textit{pre-phase space}
\be \tilde \Phi_\Sigma= \CC^\infty(\Sigma)\times \CC^\infty(\Sigma) \label{Phi tilde}\ee
with coordinates denoted by $(\phi,\phi_n)$. The pre-phase space is endowed
with the 1-form
\be \tilde\alpha_\Sigma=\int_\Sigma \mu\,(\Gamma\, \phi_n+ u(\phi))\,\delta \phi \quad \in \Omega^1(\tilde\Phi_\Sigma) \label{alpha}\ee
More concretely, for each $(\phi,\phi_n)\in\tilde \Phi_\Sigma$, $\tilde\alpha_\Sigma(\phi,\phi_n)$ is the linear map
$\tilde \Phi_\Sigma\to\mathbb{R}$,
\[
(f,f_n)\mapsto \int_\Sigma \mu\,(\Gamma\, \phi_n+ u(\phi))\, f
\]

The presymplectic structure on the pre-phase space is defined as
\be \tilde\omega_\Sigma=\delta \tilde\alpha_\Sigma=
\int_\Sigma \mu\,(\Gamma\, \delta\phi_n+ u(\delta\phi))\,\delta \phi \quad \in \Omega^2(\tilde\Phi_\Sigma) \label{presymp form}\ee
where $\delta$ in $\delta \tilde\alpha_\Sigma$ stands for de Rham differential on $\tilde\Phi_\Sigma$.\footnote{For more details on ``local'' differential forms on spaces of fields, see e.g. \cite{QFT_IAS}.}
More concretely, $\tilde\omega_\Sigma$ is the skew symmetric bilinear map
$\tilde \Phi_\Sigma\times\tilde \Phi_\Sigma\to\mathbb{R}$,
\[
((f,f_n),(g,g_n))\mapsto
\int_\Sigma \mu\,(\Gamma\, f_n+ u(f))\,g
-\int_\Sigma \mu\,(\Gamma\, g_n+ u(g))\,f
\]

The phase space $\Phi_\Sigma$ is defined as the reduction of the pre-phase space by the kernel of the presymplectic form,
$$\Phi_\Sigma=\tilde\Phi_\Sigma/ \ker(\tilde\omega_\Sigma)$$
The 2-form $\tilde\omega_\Sigma$ descends to a symplectic structure on the phase space, $\omega_\Sigma\in \Omega^2(\Phi_\Sigma)$.

\begin{remark}
The geometric data $(\Gamma,u,\mu)$ on $\Sigma$ can be considered modulo equivalence $(\Gamma,u,\mu)\sim (c\Gamma,cu,c^{-1}\mu)$ for any nonvanishing $c\in \CC^\infty(\Sigma)$. Also, the data $(\Gamma,u,\mu)$ up to this equivalence can be viewed as a section
$$(\Gamma+ u)\mu \in \Gamma(\Sigma, (\RR\oplus T\Sigma)\otimes \wedge^{\n-1}T^*\Sigma)$$
where $\RR$ stands for the trivial real line bundle over $\Sigma$.
\end{remark}

In case when $\Sigma=\dd M$ is the boundary of an $\n$-manifold $M$, the geometric data $(\Gamma,u,\mu)$ are inferred from the metric $g$ on $M$ as follows:
\be \Gamma(x) = g^{-1}(x)(n_x^*,n_x^*),\quad u(x)=g^{-1}(x)(n_x^*,\bt)-\Gamma(x) n_x\;\; \in T_x \dd M,\quad \mu = \iota_{n}\mu_g \label{geom structure from bulk}\ee
Here we chose some 
vector field on the boundary\footnote{We do not require any compatibility of $n$ with the metric on $M$.} $n\in \Gamma(\dd M, i^* TM)$ transversal to the boundary everywhere (we denote $i:\dd M\hra M$ the embedding of the boundary); $n^*\in \Gamma(\dd M, i^* T^*M)$ is the covector field on the boundary defined by $\langle n^*_x, n_x \rangle=1$, $\langle n^*_x, T_x \dd M\rangle=0 $; $\mu_g=\sqrt g\, d^\n x$ is the metric volume element on $M$; $\iota_\bt$ stands for contraction of a form with a vector field.

The projection $\tilde\pi_M: F_M\ra \tilde\Phi_{\dd M}$ sends $\phi\in \CC^\infty(M)$ to $(\phi|_{\dd M}, \dd_n \phi|_{\dd M})$ -- values of $\phi$ at the boundary and 
derivative along $n$ at the boundary.

\begin{remark} Choosing a different 
transversal
vector field at the boundary, $n'=a n+ w$ with nonvanishing $a\in \CC^\infty(\dd M)$ and with $w\in \mathfrak{X}(\dd M)$ a tangent vector field on $\dd M$, results in different induced geometric data on the boundary:
$$(\Gamma',u',\mu')=(a^{-2}\Gamma, a^{-1} u- a^{-2} w, a \mu)$$
The new projection $\tilde\pi_M':F_M\ra \tilde\Phi_\Sigma$, corresponding to $n'$, sends $\phi\in \CC^\infty(M)$ to $(\phi|_{\dd M}, \dd_{n'}\phi|_{\dd M})\in \tilde\Phi_{\dd M}$ and can be viewed as the old one, composed with a
linear isomorphism
of the pre-phase space $\tilde\Phi_{\dd M}\ra \tilde\Phi_{\dd M}$ sending $(\phi,\phi_n)\mapsto (\phi, a \phi_n+w(\phi))$.
\end{remark}

The pull-back of the 1-form $\tilde\alpha_{\dd M}$ 
to 
the space of fields $F_M$
is
\be \tilde\pi_M^* \tilde\alpha_{\dd M}=\int_{\dd M} ( *_g d \phi\cdot \delta\phi)|_{\dd M}=\int_{\dd M} ( \iota_{g^{-1}(d\phi)}\mu_g\cdot \delta \phi)|_{\dd M} \label{alpha Hodge}\ee
It arises as the boundary term of the variation of the action (\ref{S}):
$$\delta S = (-1)^{\n-1}\int_M d\delta\phi \wedge *_g d \phi = -\int_M (d *_g d \phi)\cdot \delta\phi + \tilde\pi_M^* \tilde\alpha_{\dd M} $$

\begin{remark} According to the construction of \cite{CMR, CMR2}, one associates to an $(\n-1)$-manifold $\Sigma$ with a pseudo-Riemannian metric on a 
cylinder $M_\epsilon=\Sigma\times [0,\epsilon]$ the space $\tilde\Phi_\Sigma$ of 1-jets\footnote{Only 1-jets are required since the density of the action $S$ is of second order in the field derivatives.} of functions on $M_\epsilon$ at $\Sigma\times \{0\}$. The one-form $\tilde\alpha_\Sigma\in \Omega^1(\tilde\Phi_\Sigma)$ arises as the part of the boundary term of the variation of $S$ on $M_\epsilon$ corresponding to the contribution of the boundary component $\Sigma\times \{0\}$. The geometric data $(\Gamma,u,\mu)$ introduced above constitute the part of the metric on $M_\epsilon$ necessary to define the 1-form $\tilde\alpha_\Sigma$. 
The transversal
vector field $n$ arises from the 1-jet of the embedding of the cylinder $M_\epsilon\hra M$ as a neighborhood of the boundary of $M$.
\end{remark}

\section{Two-dimensional Minkowski case}\label{sec: Minkowski plane}
Consider the Minkowski plane $\RR^{1,1}$ with coordinates $(x,y)$ and metric $g=dx^2-dy^2$. Let $D$ be a domain\footnote{By domain here we mean the closure of an open subset.} of $\RR^{1,1}$ with smooth boundary, with metric given by restriction of the Minkowski metric on $\RR^{1,1}$ to $D$. As above, $F_D=\CC^\infty(D)$ and the action (\ref{S}) is
$$S_{D,g}[\phi]=\frac12 \int_D [(\dd_x \phi)^2-(\dd_y \phi)^2]\, dx\,dy$$

Unless otherwise stated, in case of an unbounded domain $D$, we assume that $k$-th derivatives of fields have asymptotics
\be \dd^k\phi\sim O((x^2+y^2)^{-\frac{\eta+k}{2}})\label{phi asymptotics}\ee
at infinity, where $k=0,1,2,\ldots$ and $\eta>0$ is some constant.

The corresponding Euler-Lagrange equation is just the wave equation
$$\dd_x^2\phi-\dd_y^2\phi=0$$

\subsection{Examples of boundary structures}
\label{sec: examples}
In this section we consider $D$ a half-space in $\RR^{1,1}$ with space-like, time-like or light-like boundary $\Sigma=\dd D\simeq \RR$.

Consider the case $D=\RR\times [y_0,\infty)$ with space-like boundary $\dd D=\RR\times \{y_0\}$. Using the construction of Section \ref{sec: free boson}, we choose the transversal vector field at the boundary to be $n=\dd_y$ and obtain the geometric structure (\ref{geom structure from bulk}) on the boundary $(\Gamma,u,\mu)=(-1,0,-dx)$.
The pre-phase space is  $\tilde\Phi_{\dd D}=\CC^\infty(\RR)\times \CC^\infty(\RR)\ni (\phi,\phi_n)$ and the projection $\tilde\pi_D: F_D\ra \tilde\Phi_{\dd D}$ sends $\phi\in \CC^\infty(D)$ to $(\phi|_{y=y_0}, \dd_y \phi|_{y=y_0} )$.
The 1-form (\ref{alpha}) on the pre-phase space 
is
\be \tilde \alpha=\int_{\RR} dx\,\phi_n\,\delta\phi\quad \in \Omega^1(\tilde\Phi_{\dd D}) \label{spacelike alpha}\ee
and its differential
\be \tilde\omega=\int_\RR dx\, \delta \phi_n\wedge \delta \phi \quad \in \Omega^2(\tilde \Phi_{\dd D})\label{spacelike omega}\ee
is weakly non-degenerate, i.e. $\ker \tilde\omega=0$. Thus, there is no symplectic reduction and the phase space coincides with the pre-phase space, $\Phi_{\dd D}=\tilde\Phi_{\dd D}$, with symplectic structure $\omega=\tilde\omega$.

Similarly, for $D=[x_0,\infty)\times\RR$ with time-like boundary $\dd D =\{x_0\}\times \RR$ we pick $n=\dd_x$, which induces geometric data $(\Gamma,u,\mu)=(1,0,dy)$ on the boundary. The projection $\pi_D$ sends $\phi\in \CC^\infty(D)$ to $(\phi|_{x=x_0}, \dd_x \phi|_{x=x_0})$.
The 1-form $\tilde\alpha$ and its differential $\tilde \omega$ are again given by formulae (\ref{spacelike alpha}, \ref{spacelike omega}).
Again, the non-degeneracy of $\tilde\omega$ implies that $\Phi_{\dd D}=\tilde\Phi_{\dd D}$, $\omega=\tilde\omega$.

Next, consider a half-space on $\RR^{1,1}$ with light-like boundary. Using coordinates $\sigma_+=y+x$, $\sigma_-=y-x$ on $\RR^{1,1}$, we set $D=\{(\sigma_+,\sigma_-)\in \RR^{1,1}\,|\, \sigma_-\geq \sigma_-^0\}$ for some $\sigma_-^0\in\RR$. Introducing coordinate vector fields $\dd_\pm=\frac{1}{2}(\dd_y\pm \dd_x)$, we set $n=\dd_-$. This choice yields the boundary geometric data\footnote{It is useful to note that in coordinates $\sigma_\pm$, the metric, its inverse and the metric volume element on $D$ are, respectively, $g=-d\sigma_+\cdot  d\sigma_-$, $g^{-1}=-4\, \dd_+\cdot \dd_-$, $\mu_g=\frac{1}{2}\, d\sigma_+\wedge d\sigma_-$. Here $\cdot$ stands for the symmetrized tensor product.} $(\Gamma,u,\mu)=(0,-2 \dd_+, -\frac{1}{2}d \sigma_+)$, therefore
$$\tilde\alpha=\int_\RR d\sigma_+\;\dd_+ \phi\;\delta \phi ,\qquad  \tilde\omega=\delta\tilde\alpha=\int_\RR d\sigma_+\;(\dd_+ \delta\phi)\wedge\delta \phi $$
Using the linear structure on the pre-phase space, we can regard the presymplectic structure $\tilde\omega$ as an anti-symmetric bilinear form
on $\tilde\Phi$
given by
\be\tilde\omega((\phi,\phi_n),(\psi,\psi_n))=\int_\RR  d\sigma_+\;((\dd_+ \phi)\;\psi-(\dd_+ \psi)\;\phi) 
\label{omega for lightlike}\ee

The kernel of $\tilde\omega$ and hence the symplectic reduction depend on the allowed behavior of $\phi$ at $\sigma_+\ra \infty$. For instance, we have the following.
\begin{enumerate}[(i)]
\item If we require $\lim_{\sigma_+\ra \infty}\phi(\sigma_+)=0$ then the presymplectic form (\ref{omega for lightlike}) becomes
\be \tilde\omega((\phi,\phi_n),(\psi,\psi_n))=2\int_\RR d\sigma_+\; (\dd_+ \phi)\; \psi\label{omega lightlike 1}\ee
So, $(\phi,\phi_n)\in \ker\tilde\omega$ iff $\dd_+\phi=0$, but by the vanishing requirement at $\sigma_+\ra \infty$ this implies $\phi=0$. Hence, $\ker\tilde\omega=\{0\}\times \CC^\infty(\RR) \subset \tilde\Phi$ and the phase space is
\be\Phi=\tilde\Phi/\ker\tilde\omega = \CC^\infty(\RR)\ni \phi \label{phase space lightlike}\ee
with (non-degenerate) symplectic structure given by r.h.s. of (\ref{omega for lightlike}).
\item Requiring that $\phi$ has some (possibly, different) limits at $\sigma_+\ra\pm\infty$, we get a boundary term, integrating by parts in (\ref{omega for lightlike}):
$\tilde\omega((\phi,\phi_n),(\psi,\psi_n))=-|\phi\psi|^{+\infty}_{-\infty}+2\int_\RR d\sigma_+\, (\dd_+ \phi)\, \psi$.
Thus $(\phi,\phi_n)\in \ker\tilde\omega$ iff $\dd_+ \phi=0$ and $\phi(\pm\infty)=0$, which again implies $\phi=0$. So, $\ker\tilde\omega$ is the same as in case of vanishing condition at $\sigma_+\ra\infty$ and the phase space is again given by (\ref{phase space lightlike}) (though now we impose different asymptotical conditions on $\phi$).
\item Imposing periodic asymptotics $\phi(+\infty)=\phi(-\infty)$, we get back to (\ref{omega lightlike 1}) but now the kernel becomes bigger:
    $$\ker\tilde\omega=\{(\phi=C,\phi_n\in \CC^\infty(\RR))\;|\; C\in\RR\}\subset \tilde\Phi$$
    Thus the phase space is $\Phi=\CC^\infty(\RR)/\RR$ where we consider functions differing by a constant shift as equivalent. We can choose the section of this quotient e.g. by requiring $\phi(0)=0$. In this case the projection $\pi_D: F_D\ra \Phi$ maps $\phi\in \CC^\infty(D)$ to $\psi(\sigma_+)=\phi(\sigma_+,\sigma_-^0)-\phi(0,\sigma_-^0)$.
\end{enumerate}

\subsection{Canonical relations}
Related to examples of the previous section, with $D\subset \RR^{1,1}$ a half-space with boundary $\dd D=\Sigma$ a line in $\RR^{1,1}$, are cases when $D\subset \RR^{1,1}$ is a strip with boundary $\Sigma\sqcup \Sigma^\mr{op}$ where $\mr{op}$ denotes the opposite orientation. In all these cases $L_D$ is Lagrangian as we presently prove.

For $\Sigma$ space-like, consider $D=\RR\times [y_0,y_1]$. Denote $\pi:=\dd_y\phi$. Then the 1-form on the phase space
$$\Phi_{\dd D}=\CC^\infty(\RR)^{\times 4}\ni (\phi_0,\pi_0,\phi_1,\pi_1)$$
is
$$\alpha=\int_\RR (\pi_1 \delta \phi_1 - \pi_0 \delta \phi_0)\; dx$$
where subscript $i$ corresponds to boundary components $y=y_i$ of the strip, $i=0,1$ (and we are still assuming asymptotics (\ref{phi asymptotics}) for fields $\pi,\phi$). The Euler-Lagrange equation can be rewritten as a system
\begin{eqnarray*}
\dd_y \pi &=& \dd_x^2 \phi \\
\dd_y \phi &=& \pi
\end{eqnarray*}
The system is Hamiltonian with respect to the symplectic form $\int_\RR \delta\pi\wedge \delta\phi\; dx$ and to the Hamiltonian function $H=\frac{1}{2}\int_\RR (\pi^2+(\dd_x \phi)^2)\, dx$. Since $L_D$ is the graph of the corresponding Hamiltonian flow from time $y_0$ to time $y_1$, it is Lagrangian.

Similarly one proves that $L_D$ is Lagrangian for $\Sigma$ time-like, for the strip $D=[x_0,x_1]\times\RR$.

Finally, consider the case when $\Sigma$ is light-like. Passing to coordinates $\sigma_\pm$, we consider the strip $D=\{(\sigma_+,\sigma_-)\in \RR^{1,1}\, |\, \sigma_-^0\leq \sigma_-\leq \sigma_-^1\}$. The Euler-Lagrange equation becomes $\dd_+\dd_-\phi=0$, which has general solution
\be \phi(\sigma_+,\sigma_-)=f(\sigma_+)+g(\sigma_-)\label{wave eq sol}\ee with $f$ and $g$ arbitrary functions. Therefore, for any $\sigma^0_-$ and $\sigma^1_-$, $L_D$ is the diagonal in $\bar\Phi_\Sigma\times \Phi_\Sigma$, where bar denotes opposite symplectic structure, so it is Lagrangian. Observe however that $g$ cannot be determined by boundary conditions. As a consequence, on such strips we cannot have uniqueness of solutions.

\subsection{Light-like diamond}

Consider a diamond in Minkowski plane with piecewise light-like boundary,\footnote{The construction of Section \ref{sec: free boson} extends naturally to the case of manifolds with piecewise smooth boundary. In this case, for the pre-phase space (\ref{Phi tilde}) one takes pairs of piecewise smooth continuous functions (smooth where the boundary is smooth).}
$$D=\{(\sigma_+,\sigma_-)\in \RR^{1,1}\;|\; \sigma_+^0\leq \sigma_+\leq \sigma_+^1,\;\sigma_-^0\leq \sigma_-\leq \sigma_-^1\}$$
We label the four vertices of the diamond as
$$a=(\sigma_+^0, \sigma_-^0),\; b=(\sigma_+^1, \sigma_-^0),\; c=(\sigma_+^1, \sigma_-^1),\; d=(\sigma_+^0, \sigma_-^1)$$

Proceeding as in Section \ref{sec: examples}, we obtain the pre-phase space\footnote{We are not including the normal derivative $\phi_n$ in our description of $\Phi_{\dd D}$, since it does not appear in the 2-form (\ref{omega diamond}) and would be eliminated by symplectic reduction anyway.}
$$\Phi_{\dd D}=\{\phi\in C^0(\dd D)\mbox{ smooth on edges of }\dd D\}$$
We denote restrictions of $\phi$ to the four edges of the diamond by $\phi^{ab},\phi^{dc}\in \CC^\infty[\sigma_+^0,\sigma_+^1]$, $\phi^{ad},\phi^{bc}\in \CC^\infty[\sigma_-^0,\sigma_-^1]$ respectively.

The pre-symplectic 2-form induced on $\Phi_{\dd D}$ is
\be \omega=\int_{\dd D} \epsilon \;d\delta \phi\wedge\delta\phi \label{omega diamond}\ee
where $\epsilon=+1$ on two edges parallel to $\dd_+$ and $\epsilon=-1$ on the other two.
Viewed as an anti-symmetric bilinear pairing $\Phi_{\dd D}\otimes \Phi_{\dd D}\ra\RR$, the pre-symplectic structure is
\be
\omega(\phi,\psi)=2 \int_{\dd D} \epsilon\; d\phi\cdot \psi+2\;(\phi_a\psi_a-\phi_b\psi_b+\phi_c\psi_c-\phi_d\psi_d)
\label{omega diamond bilinear}\ee
where we used integration by parts to transfer derivatives from $\psi$ to $\phi$. Subscript $a,b,c,d$ stands here for evaluation of $\phi$ or $\psi$ at the corresponding vertex of the diamond.

It follows from (\ref{omega diamond bilinear}) that $\phi\in\ker\omega$ implies $\phi=C\in\RR$ -- a constant on the whole $\dd D$. On the other hand $\omega(C,\psi)=2C (\psi_a-\psi_b+\psi_c-\psi_d)$, hence $\ker\omega=0$. Thus $\omega$ is actually non-degenerate and $(\Phi_{\dd D},\omega)$ is the symplectic phase space, with no further symplectic reduction required.

\subsubsection{Evolution relation}
Using the general ansatz (\ref{wave eq sol}) for solutions of the wave equation, the evolution relation $L\subset\Phi_{\dd D}$ can be described as
$$
L=\{
\phi(\underbrace{\sigma_+,\sigma_-}_{\in\dd D})=f(\sigma_+)+g(\sigma_-)
\;|\;
f\in \CC^\infty[\sigma_+^0,\sigma_+^1],\;
g\in \CC^\infty[\sigma_-^0,\sigma_-^1]
\}
$$

To show that $L\subset \Phi_{\dd D}$ is a Lagrangian subspace (and thus verify Conjecture \ref{Conj} in this case), we check isotropicity and coisotropicity of $L$. For isotropicity, we have
\begin{multline*}
\omega|_L= \int_{[\sigma_+^0,\sigma_+^1]} d\delta f\wedge(\delta g(\sigma_-^0)-\delta g(\sigma_-^1))+
\int_{[\sigma_-^0,\sigma_-^1]} d\delta g\wedge(\delta f(\sigma_+^0)-\delta f(\sigma_+^1))=\\
=(\delta f(\sigma_+^1)-\delta f(\sigma_+^0))\wedge (\delta g(\sigma_-^0)-\delta f(\sigma_-^1))+ (\delta g(\sigma_-^1)-\delta g(\sigma_-^0))\wedge(\delta f(\sigma_+^0)-\delta f(\sigma_+^1))=0
\end{multline*}
Thus $L$ is indeed isotropic. For coisotropicity, (\ref{omega diamond bilinear}) implies that for $\phi\in L$ and $\psi$ arbitrary,
$$\omega(\phi,\psi)=2\int_{[\sigma_+^0,\sigma_+^1]} df\; (\psi^{ab}-\psi^{dc})+
2\int_{[\sigma_-^0,\sigma_-^1]} dg\; (\psi^{ad}-\psi^{bc})+
\mbox{contributions of corners}
$$
Thus $\psi\in L^\perp$ implies (by setting $f(\sigma_+^0)=g(\sigma_-^0)=0$ and taking $df$ or $dg$ to be the difference of two bump 1-forms localized near two points, so that the total integral vanishes) $\psi^{ab}-\psi^{dc}=C$, $\psi^{ad}-\psi^{bc}=C'$ where $C,C'\in\RR$ are two constants. This implies in turn that $\psi\in L$, with corresponding $f_\psi(\sigma_+),g_\psi(\sigma_-)$ given by
$$f_\psi(\sigma_+)=\psi^{ab}(\sigma_+)-\psi^{ab}(\sigma_+^0),\quad g_\psi(\sigma_-)=\psi^{ad}(\sigma_-)$$
This proves coisotropicity of $L$ and hence $L$ is indeed Lagrangian.

\subsubsection{Hamilton-Jacobi action}
Restriction of the action (\ref{S}) to solutions of Euler-Lagrange equation is in general
$$S|_{EL}=\frac12 \int_M d\phi\wedge *d\phi=\underbrace{-\frac12 \int_M \phi\wedge d*d\phi}_{0\;\mr{on}\; EL}+\frac12 \int_{\dd M}(\phi\wedge*d\phi)|_{\dd M}$$
Since this expression is given by a boundary term, it descends to a function on $L$ (at least as a subspace of the pre-phase space, in the general case).

In case of the diamond we have
$$
S|_{EL}=\frac12
\int_{\dd D}\epsilon\; \phi\,d\phi
=\frac12\, (-\phi_a^2+ \phi_b^2-\phi_c^2+\phi_d^2)\quad \in \CC^\infty(L)
$$

Note 
that this Hamilton-Jacobi action depends only on the values of $\phi$ at the vertices of the diamond.

\section{Wave equation on compact domains in Minkowski plane}
\label{sec: compact}

Let $D\subset \mathbb{R}^{1,1}$ be a connected compact domain in the Minkowski plane. We make the following assumptions about its boundary $\gamma=\dd D$.
\begin{enumerate}[(A)]
\item 
Each connected component $\gamma_k$, $1\leq k\leq N$, of the boundary $\gamma$ is a smooth simple closed curve.
\label{assump A}
\item There are finitely many points on $\gamma$ with light-like 
tangent; we denote this set of points $I$. \label{assump B}
\item 
The curvature of $\gamma$ (as a multi-component plane curve) at points of $I$ is non-zero.
\label{assump C}
\end{enumerate}


Assume that each curve $\gamma_k$ is parameterized by $t\in \mathbb{R}/(T_k\cdot \mathbb{Z})$, with $T_k\in \mathbb{R}$ the period. We assume that the orientation of $\gamma_k$ induced from the parametrization agrees with the one induced from the orientation of $D$. Define $\theta: \gamma_k \ra \RR/(\pi\cdot \ZZ)$ and $v:\gamma_k \ra \RR_{>0}$ by
$\theta(t)=\arctan(\frac{\dot y}{\dot x})+\frac{\pi}{2}$, $v(t)=(\dot x^2+\dot y^2)^{1/2}$.

\subsection{Phase space, symplectic structure.}
The phase space\footnote{In the terminology of Section \ref{sec: free boson}, we should be calling it the pre-phase space. Below (cf. Proposition \ref{prop: omega non-deg}) we will show that the presymplectic form on $\Phi_\gamma$ is in fact symplectic, so that no further symplectic reduction is needed. Thus the terminology is justified.} (the space of boundary fields) associated to $\gamma$ is $\Phi_\gamma=\{(\phi,\phi_n)\in \CC^\infty(\gamma)\times \CC^\infty(\gamma)\}$
The projection $\pi: F_D\ra \Phi_\gamma $ sends $\phi\in \CC^\infty(D)$ to its restriction to $\gamma$ and the normal derivative at a point on $\gamma$; ``normal'' means an outward pointing unit normal vector to the boundary
with respect to
\textit{Euclidean} metric on the plane.

The geometric data (\ref{geom structure from bulk}) on $\gamma$, associated to the choice of the Euclidean normal vector field $n=\cos\theta\,\dd_x+\sin\theta\, \dd_y$, is:
$(\Gamma,u,\mu)=(\cos(2\theta),-\frac1v\sin(2\theta)\dd_t, v\, dt)$, which yields the following boundary 1-form (\ref{alpha})
on $\Phi_\gamma$:
$$\alpha=\int_\gamma dt\; (v \cos(2\theta) \dd_n-\sin(2\theta)\dd_t)\phi\; \delta\phi$$
where $\dd_n\phi := \phi_n$ is a notation.
It generates a constant 2-form on $\Phi_\gamma$
$$\omega=\delta \alpha = \int_\gamma dt\;(v \cos(2\theta) \dd_n-\sin(2\theta)\dd_t)\delta\phi\;\wedge \delta\phi$$
Using the linear structure on $\Phi_\gamma$, we can view $\omega$ as an anti-symmetric pairing $\Phi_\gamma\otimes \Phi_\gamma \ra \RR$,
\begin{multline}\label{omega as a pairing}
\omega((\phi,\phi_n),(\psi,\psi_n))=\\
=\int_\gamma dt\cdot \left( v\cos(2\theta) \phi_n\; \psi-\sin(2\theta) \dd_t \phi\; \psi - \phi\; v\cos(2\theta)\psi_n + \phi\;\sin(2\theta)\dd_t \psi \right)
\end{multline}

\begin{proposition}\label{prop: omega non-deg}
Two-form $\omega$ is non-degenerate on $\Phi_\gamma$.
\end{proposition}
\begin{proof}
Indeed, by (\ref{omega as a pairing}), a pair $(\phi,\phi_n)\in \Phi_\gamma$ is in the kernel of $\omega$ if and only if
$$ \left\{\begin{array}{l} -v\cos(2\theta) \phi = 0  \\ -\dd_t(\sin(2\theta) \phi)+v\cos(2\theta)\phi_n-\sin(2\theta) \dd_t \phi=0   \end{array} \right. \Leftrightarrow \left\{ \begin{array}{l} \phi=0 \\ \phi_n=0  \end{array} \right. $$
where we use that, by assumption (\ref{assump B}), $\cos(2\theta)$ vanishes in isolated points.
\end{proof}

\subsection{Evolution relation: main theorem.}
Set $EL_D=\{\phi\in \CC^\infty(D)\;|\; d*d\phi=0\}\subset F_D$
-- the space of solutions to the wave equation in $D$ and also set
$$L=\pi(EL_D)\subset \Phi_\gamma$$
-- the \textit{evolution relation}.

\begin{thm}\label{thm}
The evolution relation $L$ is a Lagrangian subspace of $\Phi_\gamma$. 
\end{thm}

\subsubsection{Evolution relation in the simply connected case and involutions $E_\pm$ on the boundary
}
\label{sec: Evol relation and involutions}
In case when $D$ is \textit{simply connected} (N=1), the space of solutions of the wave equation in the bulk $EL_D
$ is given by
\be EL_D=\{\phi=F+G\;|\; F,G\in \CC^\infty(D),\; \dd_- F=\dd_+G=0 \} \label{EL simply connected}\ee
Note that globally $\dd_- F=0$ does not imply $F=F(\sigma_+)$, e.g. if $D$ is not convex.


The two 
distributions
$\dd_\pm$ on $D$ induce two equivalence relations $\epsilon_\pm$ on points of $D$, where two points in $D$ are considered equivalent if they can be connected by a light-like segment with tangent $\dd_{\pm}$ lying inside $D$. In turn, $\epsilon_\pm$ induce equivalence relations $\mathcal{E}_\pm$ on points of $\gamma$.

Denote $I_{\pm}=\{p\in I\, |\, \theta(p)=\mp \pi/4\}$, so that $I=I_+\sqcup I_-$.

By assumptions (\ref{assump A},\ref{assump B}), an equivalence class of $\EE_\pm$ of order $1$ is necessarily a point of $I_{\pm}$ and an equivalence class of order $n\geq 3$ necessarily contains $n-2$ points of $I_\pm$.
Thus there is only a finite set of points $I'_\pm\subset \gamma$ with equivalence class of $\EE_\pm$ of order $\neq 2$.

Therefore, equivalence relations $\EE_\pm$ induce two orientation-reversing smooth involutions $E_\pm : (\gamma-I'_\pm) \ra (\gamma-I'_\pm)$,
i.e. for a point $p\in\gamma-I'$, $E_\pm(p)$ is the point on $\gamma$ where one of the two light-like lines in $D$ starting at $p$ hits $\gamma$ second time.\footnote{The reader is referred to Section \ref{sec: examples of E_pm} for explicit formulae for $E_\pm$ in some examples.}

Denote $$\CC^\infty(\gamma)^{E_{\pm}}=\{f\in \CC^\infty(\gamma)\,|\, f\circ E_{\pm}=f \mbox{ on }\gamma-I'_\pm \}$$

To describe the evolution relation $L_{D}=\pi(EL_D)\subset \Phi_\gamma$, we need the following two decompositions for the unit (Euclidean) normal vector $\dd_n$ at a point on $\gamma$:
\begin{eqnarray*}\dd_n &=& -\frac{1}{v} \cot(\theta-\pi/4) \dd_t + \sqrt{2} \frac{1}{\sin(\theta-\pi/4)} \dd_- \\
\dd_n &=& -\frac{1}{v} \cot(\theta+\pi/4) \dd_t + \sqrt{2} \frac{1}{\sin(\theta+\pi/4)} \dd_+
\end{eqnarray*}
If we denote $f=F|_\gamma, g=G|_\gamma \in \CC^\infty(\gamma)$, then $\dd_- F=\dd_+ G=0$ implies
$$ \dd_n F= -\frac{1}{v} \cot (\theta-\pi/4)\; \dd_t f,\qquad \dd_n G= -\frac{1}{v} \cot (\theta+\pi/4)\; \dd_t g $$
Thus, for $D$ \textit{simply connected}, we may describe $L$ as
\begin{multline}\label{L}
L=\{(\phi,\phi_n)=\left(f+g, -\frac{1}{v} (\cot(\theta-\pi/4)\; \dd_t f+ \cot(\theta+\pi/4)\; \dd_t g ) \right) \; |
\\ |
\;  f\in \CC^\infty(\gamma)^{E_-},\; g\in \CC^\infty(\gamma)^{E_+}\}
\end{multline}
Note that for this description we implicitly use the property that the maps
$$\{F\in \CC^\infty(D)\; |\; \dd_\pm F =0 \} \quad \xrightarrow{\pi} \quad  \CC^\infty(\gamma)^{E_\pm} $$
are surjective, for which assumption (\ref{assump C}) is essential. Note also that the expression $\cot(\theta-\frac{\pi}{4}) \dd_t f+ \cot(\theta+\frac{\pi}{4}) \dd_t g $ in (\ref{L}) is smooth on the whole $\gamma$.

\subsubsection{Evolution relation in the non-simply connected case}
\label{sec: Evol rel non-simply connected}
In general, when $D$ is not necessarily simply connected, the r.h.s. of (\ref{EL simply connected}) is valid as a \textit{local} description of the space of solutions, but globally $F,G$ may fail to exist as single valued functions on $D$. One global description of $EL_D$ is as follows:
\be EL_D=\{\phi\in \CC^\infty(D)\;|\; d\phi=\kappa+\lambda,\; \mr{where}\; \kappa,\lambda\in \Omega^1_\mr{closed}(D),\, \iota_{\dd_-}\kappa= \iota_{\dd_+}\lambda=0 \}  \label{EL non simply connected}\ee
where $\iota_{\dd_\pm}$ is the contraction with the vector field $\dd_\pm$.

For $D$ non-simply connected (note that the involutions $E_\pm$ still make perfect sense, though now they may relate pairs of points in different connected components of $\gamma$), the r.h.s. of (\ref{L}) defines a subspace $L^\mr{glob}\subset L$ corresponding to solutions of the wave equation with single valued $F,G$:
$L^\mr{glob}=\pi(EL_D^\mr{glob})$ where
$EL_D^\mr{glob}$ is given by r.h.s. of (\ref{EL simply connected}).
\begin{lemma}
\be \dim (L/L^\glob)=N-1 \label{L/L^glob dim}\ee
\end{lemma}
\begin{proof}
In $D$ we have a short exact sequence
\be EL^\mr{glob}_D \hra EL_D \twoheadrightarrow H^1(D) 
\label{EL^glob seq}\ee
where $H^1(D)$ is the de Rham cohomology of $D$ in degree 1; the second arrow sends $\phi\mapsto [\kappa]\in H^1(D)$ where we use description (\ref{EL non simply connected}). Surjectivity of the second map follows from surjectivity of the map $\{\kappa\in \Omega^1_\mr{closed}(D)\,|\, \iota_{\dd_-}\kappa=0\}\ra H^1(D)$ sending $\kappa\mapsto [\kappa]$. To prove the latter, note that we can reorder boundary components so that for any $1\leq i< N$ there exists an open subset $U_i\subset\gamma_i-I'$ such that $E_-(U_i)\subset\gamma_j$ for some $j>i$. For every $i$, take $\psi_i\in \Omega^1(U_i)$ a bump 1-form supported on $U_i$, and construct a closed $\dd_-$-horizontal 1-form on $D$ as
$\kappa_i=p_-^*(\psi_i+E_-^*\psi_i) $ where $p_-: D\ra \gamma/E_-$ is the projection to the boundary along $\dd_-$. It easy to see, by looking at periods along $\gamma_i$, that restrictions to the boundary $\{\kappa_i|_\gamma\}_{i=1}^{N-1}$ span the kernel of $H^1(\gamma)\ra \RR$ (pairing with the fundamental class of $\gamma$). Therefore $\{\kappa_i\}$ span $H^1(D)$.


It follows from (\ref{EL^glob seq}) that $\dim(EL_D/EL_D^\glob)=N-1$ and since $\pi: EL_D\ra L$ is an isomorphism\footnote{Surjectivity 
follows from the
definition of $L$. Injectivity can be seen as follows: restrictions to $\gamma$ of the 1-forms $\kappa$, $\lambda$ of (\ref{EL non simply connected}) can be explicitly and uniquely recovered from $(\phi,\phi_n)\in L$ by formulae (\ref{coiso_alpha},\ref{coiso_beta}) below. Hence $(\phi,\phi_n)=0$ implies $\kappa|_\gamma=\lambda|_\gamma=0$, which in turn implies, by $\dd_\mp$-horizontality of $\kappa,\lambda$, that $\kappa=\lambda=0$ in $D$. Hence, there can be no non-zero point of $EL_D$ inducing zero on the boundary.}, we have
$ \dim (L/L^\glob)=N-1$.
\end{proof}

\subsection{Proof of Theorem \ref{thm}.}

\begin{lemma}\label{lemma L iso}
$L\subset \Phi_\gamma$ is isotropic.
\end{lemma}

\begin{proof}
Indeed, due to (\ref{alpha Hodge}) and using Stokes' theorem, for $(\phi,\phi_n),(\psi,\psi_n)\in L$ we have
\begin{multline*}
\omega((\phi,\phi_n),(\psi,\psi_n))=\int_\gamma ((* d\tilde\phi)\;\tilde\psi -  (* d\tilde\psi)\;\tilde\phi)|_\gamma= \\
=\int_D d( (* d\tilde\phi)\;\tilde\psi - (* d\tilde\psi)\;\tilde\phi)=
\int_D  (d*d \tilde\phi)\;\tilde\psi - (d*d \tilde\psi)\;\tilde\phi =0
\end{multline*}
where $\tilde\phi,\tilde\psi\in EL_D$ are extensions of $(\phi,\phi_n),(\psi,\psi_n)$ into the bulk $D$ as solutions of the wave equation.
\end{proof}
This proof is a specialization of a general argument, applicable to any classical field theory, cf. \cite{CMR}.

Note that Lemma \ref{lemma L iso} implies that $L^\glob$ is isotropic in $\Phi_\gamma$.
\begin{lemma}\label{lemma L^glob dim}
$$\dim \frac{(L^\glob)^\perp}{L^\glob}=2\,(N-1)$$
\end{lemma}

\begin{proof}
Let us calculate the symplectic complement of $L^\glob$ in $\Phi_\gamma$. For $(\psi,\psi_n)\in L^\glob$, with $f,g$ denoting the $E_\mp$-invariant parts as in (\ref{L}), we have
\begin{multline}\omega ((\phi,\phi_n),(\psi,\psi_n))= -\int_\gamma dt \;\left(\phi\; \dd_t (f-g) - (f+g)\; (v \cos(2\theta) \phi_n-\sin(2 \theta) \dd_t \phi)\right)
\\=
-\int_\gamma dt\;f\; (-(1-\sin(2\theta))\dd_t \phi - v\cos(2\theta) \phi_n) -
\int_\gamma  dt\;g\; ((1+\sin(2\theta))\dd_t \phi - v\cos(2\theta) \phi_n)
\label{coiso_omega}
\end{multline}
Therefore $(L^\glob)^\perp$ consists of pairs $(\phi,\phi_n)\in \Phi_\gamma$ for which the 1-forms
\begin{eqnarray}\label{coiso_alpha} \alpha&=&-\frac{1}{2}\;dt\; (-(1-\sin(2\theta))\dd_t \phi - v\cos(2\theta) \phi_n) \in \Omega^1(\gamma),\\ \label{coiso_beta}
\beta&=&\frac12\;dt\; ((1+\sin(2\theta))\dd_t \phi - v\cos(2\theta) \phi_n)\in \Omega^1(\gamma)
\end{eqnarray}
are $E_-$- and $E_+$-invariant, respectively.

The inverse of (\ref{coiso_alpha},\ref{coiso_beta}) is given by
$$d\phi= \alpha+\beta,\quad dt\; \phi_n=-\frac{1}{v}\left( \cot(\theta-\pi/4)\,\alpha+ \cot(\theta+\pi/4) \,\beta \right)$$

The map $\rho:\Phi_\gamma \ra \Omega^1(\gamma)\times \Omega^1(\gamma)$ sending $(\phi,\phi_n)\mapsto (\alpha,\beta)$, as defined by (\ref{coiso_alpha},\ref{coiso_beta}), has image
\be \mr{im}(\rho)=\{(\alpha,\beta)\in \Omega^1(\gamma)\times \Omega^1(\gamma)\;|\; \alpha+\beta\in \Omega^1_\mr{exact}(\gamma), \; \alpha\; \mr{vanishes\;on}\; I_-, \; \beta\; \mr{vanishes\; on}\; I_+\}\label{coiso_im}\ee
and kernel
\be \ker(\rho)=\{(\phi,\phi_n)\in \Omega^0_\mr{closed}(\gamma)\times \{0\}\} \label{coiso_ker}\ee


On the other hand, the value of $\rho$ on $(\phi,\phi_n)\in L^\glob$ is $(\alpha,\beta)=( df,  dg)$, where $f,g$ are the $E_\mp$-invariant parts of $\phi$ as in r.h.s. of (\ref{L}). Thus for the restriction of $\rho$ to $L^\glob$ we have
\be \mr{im}(\rho|_{L^\glob})= \Omega^1_\mr{exact}(\gamma)^{E_-}\times \Omega^1_\mr{exact}(\gamma)^{E_+} \label{coiso rho(L)}\ee
and the kernel is
$$\ker(\rho|_{L^\glob})= \ker(\rho)\cap L^\glob= \{(\phi,\phi_n)=(C,0)\;|\; C\in \RR\}$$

By 
(\ref{coiso_omega}),
\be (L^\glob)^\perp=\rho^{-1}(\Omega^1(\gamma)^{E_-}\times \Omega^1(\gamma)^{E_+}) \label{(L^glob)^perp}\ee
in particular, due to (\ref{coiso_im}),
\be \mr{im}(\rho|_{(L^\glob)^\perp})= \{(\alpha,\beta)\in \Omega^1(\gamma)^{E_-}\times \Omega^1(\gamma)^{E_+}\;|\; \alpha+\beta\in \Omega^1_\mr{exact}(\gamma)\} \label{coiso rho(L^perp)}\ee
and $\ker(\rho|_{(L^\glob)^\perp})=\ker(\rho)$, cf. (\ref{coiso_ker}). Therefore, the quotient $(L^\glob)^\perp/L^\glob$ fits into the short exact sequence
\be \Omega^0_\mr{closed}(\gamma)/\{\mr{constants}\}\hookrightarrow (L^\glob)^\perp/L^\glob \stackrel{\rho}{\twoheadrightarrow}
\rho((L^\glob)^\perp)/\rho(L^\glob) \label{coiso SES}\ee
The space on the left here is $(N-1)$-dimensional. To find $\dim((L^\glob)^\perp/L^\glob)$, we need to find the dimension of the space on the right.

Define the map $\sigma: \Omega^1(\gamma)\times \Omega^1(\gamma)\ra \RR^{2N}$ sending two 1-forms on $\gamma$ to the set of their periods around the connected components of $\gamma$,
$$(\alpha,\beta)\mapsto \left(\oint_{\gamma_1} \alpha,\cdots ,\oint_{\gamma_N} \alpha, \oint_{\gamma_1} \beta,\cdots, \oint_{\gamma_N} \beta\right)$$
The kernel of $\sigma$ is $\ker(\sigma)=\Omega^1_\mr{exact}(\gamma)\times \Omega^1_\mr{exact}(\gamma)$. Note that by (\ref{coiso rho(L)},\ref{coiso rho(L^perp)}), this implies
$\ker(\sigma)\cap \rho((L^\glob)^\perp)=\rho(L^\glob)$. Thus $\sigma$ induces an injective map $\sigma: \rho((L^\glob)^\perp)/\rho(L^\glob)\hookrightarrow \RR^{2N}$. Its image is
\begin{multline} \label{coiso im(sigma)}\sigma(\rho((L^\glob)^\perp)/\rho(L^\glob))=\\
=\{(a_1,\ldots,a_N,b_1,\ldots, b_N)\;|\; \sum_{i=1}^N a_i = \sum_{i=1}^N b_i = 0,\; a_1+b_1=0,\ldots, a_N+b_N=0 \} \end{multline}
Here the relations $\sum_i a_i = \sum_i b_i =0$ arise because $\int_\gamma \alpha =0$ for $\alpha\in \Omega^1(\gamma)^{E_\pm}$, since the involutions $E_\pm$ are orientation-reversing. The relations $a_i+b_i=0$ arise because of the relation $\alpha+\beta\in \Omega^1_\mr{exact}(\gamma)$ in (\ref{coiso rho(L^perp)}). The dimension of the right hand side of (\ref{coiso im(sigma)}) is $2N-(N+2)+1 = N-1$ (since there are $N+2$ relations and one relation between relations, $\left(\sum_i a_i\right) + \left(\sum_i b_i\right) - \sum_i (a_i+b_i)$=0). Hence, $\dim \rho((L^\glob)^\perp)/\rho(L^\glob)=N-1$ and, by (\ref{coiso SES}), $\dim((L^\glob)^\perp/L^\glob)=2\,(N-1)$.

\end{proof}

\begin{lemma} \label{lemma L^glob symp}
\begin{enumerate}[(i)]
\item \label{lemma L^glob symp (i)} The quotient $(L^\glob)^\perp/L^\glob$ inherits a non-degenerate symplectic pairing from $\Phi_\gamma$.
\item \label{lemma L^glob symp (ii)} The symplectic double orthogonal to $L^\glob$ in $\Phi_\gamma$ is $(L^\glob)^{\perp\perp}=L^\glob$.
\end{enumerate}
\end{lemma}

\begin{proof}
It follows from the proof of Lemma \ref{lemma L^glob dim} that $(L^\glob)^\perp/L^\glob$ fits into the following exact sequence:
\be \RR\ra H^0(\gamma)\ra (L^\glob)^\perp/L^\glob \ra H^1(\gamma)\ra \RR \label{L^perp/L LES}\ee
Here the maps, going from left to right, are:
\begin{itemize}
\item realization of constants as constant functions on $\gamma$,
\item realization of locally constant functions on $\gamma$ as elements of $(L^\glob)^\perp$ (with vanishing $\phi_n$),
\item map $\sigma\circ\rho : (L^\glob)^\perp/L^\glob \ra H^1(\gamma)\times H^1(\gamma)$ composed with projection to the first factor,
\item pairing with fundamental class of $\gamma$.
\end{itemize}
The symplectic structure $\omega$ on $\Phi_\gamma$ induces a well defined pairing $\underline{\omega}$ on $(L^\glob)^\perp/L^\glob$.
Using the truncation of sequence (\ref{L^perp/L LES})
\be H^0(\gamma)/\RR\ra (L^\glob)^\perp/L^\glob \ra  H^1(\gamma)|_{\int_\gamma=0} \label{L^perp/L SES}\ee
and the fact that symplectic structure (\ref{omega as a pairing}) can be written as
$$\omega((\phi,\phi_n),(\psi,\psi_n))=\int_\gamma  \phi (-\alpha_\psi+\beta_\psi)-\psi (-\alpha_\phi+\beta_\phi)$$
we see that, choosing some splitting of (\ref{L^perp/L SES}) from the right, we can write the block matrix of $\underline\omega$ as
\be \left(\begin{array}{ll} 0 & -2\langle,\rangle \\ 2\langle,\rangle &  \ast \end{array}\right) \label{L^perp/L omega}\ee
where the first and second row/column correspond to the left and right terms of (\ref{L^perp/L SES}) respectively; $\langle,\rangle$ is the non-degenerate pairing between the left and right terms of (\ref{L^perp/L SES}) induced from Poincar\'e duality $H^0(\gamma)\otimes H^1(\gamma)\ra \RR$; the lower right block is dependent on the choice of splitting of (\ref{L^perp/L SES}). Ansatz (\ref{L^perp/L omega}) implies that the anti-symmetric pairing $\underline\omega$ on $L^\perp/L$ is \textit{non-degenerate}. Thus $(L^\glob)^\perp/L^\glob$ is the symplectic reduction of $L^\perp$ and $\underline\omega$ is the induced symplectic structure on reduction. Non-degeneracy of $\underline\omega$ also immediately implies that $(L^\glob)^{\perp\perp}=L^\glob$.

\end{proof}

\textit{Proof of theorem \ref{thm}.}
The map $\rho: F_\gamma\ra \Omega^1(\gamma)\times \Omega^1(\gamma)$ defined in the proof of Lemma \ref{lemma L^glob dim} sends $(\phi,\phi_n)\in L$ to $(\, \kappa|_\gamma, \, \lambda|_\gamma)$, where $\kappa,\lambda$ are closed $\dd_\mp$-horizontal 1-forms corresponding to $(\phi,\phi_n)$ by (\ref{EL non simply connected}). Thus the image of $\rho$ on $L$ is
\be
\rho(L)= \{(\alpha,\beta)\in \Omega^1(\gamma)^{E_-}\times\Omega^1(\gamma)^{E_+}\;|\; \alpha+\beta\in \Omega^1_\mr{exact}(\gamma)\}
\label{rho(L)}\ee
Hence, by (\ref{(L^glob)^perp}), $L\subset (L^\glob)^\perp$. Taking into account isotropicity of $L$, we have a sequence of inclusions
$$L^\glob\subset L\subset L^\perp \subset (L^\glob)^\perp$$
Passing to the symplectic reduction (quotient by $L^\glob$) we get
$$L/L^\glob \subset L^\perp/L^\glob \subset (L^\glob)^\perp/L^\glob$$
By (\ref{L/L^glob dim}) and Lemma \ref{lemma L^glob dim}, $L/L^\glob$ is an $(N-1)$-dimensional isotropic subspace in a $2\,(N-1)$-dimensional symplectic space, hence $L/L^\glob$ is Lagrangian. Hence, $L/L^\glob=L^\perp/L^\glob$ and therefore $L=L^\perp$. This finishes the proof that $L$ is Lagrangian.
$\Box$

\section{Remarks.} \label{sec: remarks}
Unless stated otherwise, in this Section we are assuming the setup of Section \ref{sec: compact}.
\subsection{Dirichlet polarization}
It is interesting that $\dim((L^\glob)^\perp/L^\glob)$ depends only on the \textit{topology} of the domain $D$, at least as long as the mild assumptions \ref{assump A}, \ref{assump B}, \ref{assump C} hold. On the other hand, $L$ itself is sensitive to the \textit{geometry} of the boundary $\gamma$, in particular to dynamics on points of $\gamma$ defined by joint action of involutions $E_+,E_-$. In particular, for the map $\DD: L\ra \CC^\infty(\gamma)$, sending $(\phi,\phi_n)\mapsto \phi$, we have the following (we assume for simplicity that $D$ is simply connected).
    \begin{itemize}
    \item If there is a point on the boundary $p\in \gamma$ and a number $n\geq 1$ such that
    \be (E_+ E_-)^n p =p \label{periodic orbit} \ee
    then by (\ref{L}) on $L$ we have $\sum_{i=0}^{n-1}\phi((E_+ E_-)^i p)-\phi(E_-(E_+ E_-)^i p)=0$, hence $\DD$ is not surjective (equivalently, in general there is no existence for Dirichlet boundary problem for the wave equation on $D$).
    \item If there is an open subset of the boundary $U\subset \gamma- I$ such that (\ref{periodic orbit}) holds for every $p\in U$ for some fixed $n\geq 1$, then $\DD$ is not injective (no uniqueness for Dirichlet problem): for $\psi_U$ a bump function supported on $U$, we define
        $$f=\sum_{i=0}^{n-1} \left((E_+^* E_-^*)^i \psi_U +  E_-^* (E_+^* E_-^*)^i \psi_U \right) =
        \sum_{i=0}^{n-1} \left((E_-^* E_+^*)^i \psi_U +  E_+^* (E_-^* E_+^*)^i \psi_U \right)$$
        Then $f$ is simultaneously $E_+$- and $E_-$-invariant, hence by (\ref{L}),
        $$(0, -\frac{1}{v}(\cot(\theta-\pi/4)-\cot(\theta+\pi/4))\;\dd_t f)\in L$$
        is a non-zero vector in $L$ lying in kernel of $\DD$.
    \item If there is a point $p\in \gamma$, such that its orbit under the joint action of $E_+$ and $E_-$ is \textit{dense} in $\gamma$, then $\DD$ is injective (there is uniqueness for Dirichlet problem):
        by (\ref{L}), to have a vector in $L$ lying in kernel of $\DD$, we need a function $f\in \CC^\infty(\gamma)$ which is both $E_+$- and $E_-$-invariant. But $f$ has to be constant on the dense $E_\pm$-orbit in $\gamma$, thus $f$ is a constant and gives zero vector in $L$.
    \end{itemize}

\subsection{Explicit examples of involutions $E_\pm$: disk and annulus}\label{sec: examples of E_pm}
First consider a unit disk on $\RR^{1,1}$, defined in polar coordinates $x=r \cos\theta$, $y=r\sin\theta$ by $r\leq 1$ with the boundary unit circle parameterized by the angular coordinate $t=\theta\in \RR/(2\pi\ZZ)$.\footnote{Note that this convention agrees with conventions introduced in the beginning of Section \ref{sec: compact}, but now $\theta$ is to be considered modulo $2\pi$, not modulo $\pi$.}
The four light-like points on the boundary are: $$I=\{\underbrace{\pi/4,-3\pi/4}_{I_-},\underbrace{-\pi/4,3\pi/4}_{I_+}\}$$
and the involutions $E_\pm$ on the boundary circle are:
$$E_-: \theta\leftrightarrow \pi/2-\theta,\qquad E_+: \theta\leftrightarrow -\pi/2-\theta$$

Next, consider the annulus defined by $r_1\leq r\leq r_2$. We consider both inner and outer circle parameterized by the angular coordinate $\theta$. We will put superscripts ``in'', ``out'' to indicate to which boundary component a point belongs. The eight light-like boundary points are:
$$I=\{\underbrace{(\pi/4)^\In,(-3\pi/4)^\In,(\pi/4)^\Out,(-3\pi/4)^\Out}_{I_-}, \underbrace{(-\pi/4)^\In,(3\pi/4)^\In,(-\pi/4)^\Out,(3\pi/4)^\Out}_{I_+}\}$$
The involutions are:
\be \label{annulus involutions 1}
E_\pm: \theta^\In\leftrightarrow \left(\mp\frac{\pi}{4}+\arccos \left(\frac{r_1}{r_2}\cos(\theta\pm \frac\pi4)\right)\right)^\Out \ee
\be \label{annulus involutions 2}
E_\pm: \theta^\Out\leftrightarrow \left(\mp \frac{\pi}{2}-\theta\right)^\Out \mbox{ for }\\
 \theta^\Out\in \left(\mp\frac\pi4-\theta_0,\mp\frac\pi4+\theta_0\right)^\Out\cup \left(\pm\frac{3\pi}4-\theta_0,\pm\frac{3\pi}4+\theta_0\right)^\Out \ee
where $\theta_0=\arccos\frac{r_1}{r_2}$ and the sign of $\arccos$ in (\ref{annulus involutions 1}) is chosen in such a way that in the limit $r_1\ra r_2$ we get the 
involution
$\theta^\In\leftrightarrow \theta^\Out$.
For each choice of the sign $\pm$, the equivalence relation $\EE_\pm$ has two equivalence classes of order 1:
$\{(\mp\pi/4)^\Out\}$, $\{(\pm 3\pi/4)^\Out\}$ and two equivalence classes of order 3:
$$\{\left(\mp\frac\pi4\right)^\In, \left(\mp\frac\pi4-\theta_0\right)^\Out, \left(\mp\frac\pi4+\theta_0\right)^\Out\}, \quad
\{\left(\pm\frac{3\pi}4\right)^\In, \left(\pm\frac{3\pi}4-\theta_0\right)^\Out, \left(\pm\frac{3\pi}4+\theta_0\right)^\Out\}$$
Elements of the latter classes correspond to points of the boundary where involution $E_\pm$ is discontinuous. All the other equivalence classes are of order 2.

\subsection{Constraint (Cauchy) subspace of the phase space}
\label{sec: C}
Fix a closed curve $\gamma\subset \RR^{1,1}$ subject to assumptions \ref{assump A}, \ref{assump B}, \ref{assump C} of Section \ref{sec: compact}. Denote $D_\mr{in}$ the compact domain of $\RR^{1,1}$ bounded by $\gamma$ and denote $D_\mr{out}$ the complement of $D_\mr{in}$ in $\RR^{1,1}$.

By specializing a general construction of \cite{CMR2}, one can associate to $\gamma$  two 
subspaces of the phase space  $C_\mr{in},C_\mr{out}\subset \Phi_\gamma$ consisting of pairs $(\phi,\phi_n)\in \CC^\infty(\gamma)^{\times 2}$ extendable as solutions of the wave equation into some open neighborhood of $\gamma$ in $D_\mr{in}$ or $D_\mr{out}$ respectively. Note that one can view $C_\mr{in}$, $C_\mr{out}$ as being associated to the two orientations of $\gamma$:
$C_\mr{in}=C(\gamma)\subset (\Phi_\gamma,\omega_\gamma)$, $C_\mr{out}=C(\gamma^\mr{op})\subset (\Phi_{\gamma^\mr{op}}=\Phi_\gamma,\omega_{\gamma^\mr{op}}=-\omega_\gamma)$, where $\gamma$ is understood as coming with counterclockwise orientation by default and ``op'' denotes orientation reversal.

\begin{remark}
A related concept to the Cauchy subspaces $C_\In$, $C_\Out$ introduced above is the subspace $\mathsf{C}$ of $\Phi_\gamma$ consisting of pairs $(\phi,\phi_n)$ extendable as solutions of the wave equation into a tubular neighborhood of $\gamma$ in $\RR^{1,1}$ (as opposed to an open neighborhood in the relative topology of $D_\In$ or $D_\Out$). Obviously, $\mathsf{C}=C_\In\cap C_\Out$.
\end{remark}

We split light-like points of $\gamma$ into those where $D_\mr{in}$ is convex and those where $D_\mr{out}$ is convex: $I=I^\mr{in}\sqcup I^\mr{out}$. We also introduce involutions $E_\pm^\mr{in}$, $E_\pm^\mr{out}$ on points of $\gamma$, induced by following light-like lines in $D_\mr{in}$ or $D_\mr{out}$ respectively. Note that since $D_\mr{out}$ is non-compact, a light-like line starting at a point on $\gamma$ may run to infinity, thus involutions $E_\pm^\mr{out}$ are only defined on some subsets of $\gamma$. In particular, $E_\pm^\mr{out}$ is defined in an open neighborhood of points of $I^\mr{out}_\pm$ (with the same $\pm$).

\begin{proposition}\label{prop: C}
\begin{enumerate}
\item
The subspace $C_\mr{in}\subset \Phi_\gamma$ consists of pairs $(\phi,\phi_n)\in \CC^\infty(\gamma)^{\times 2}$
such that:
\begin{enumerate}
\item \label{C1} For every point $z\in I_\pm^\mr{in}$ there is an $E_\pm^\mr{in}$-invariant open neighborhood $z\in U_z\subset\gamma$ such that the restriction of the 1-form $\rho_\pm\in \Omega^1(\gamma)$ to $U_z$ is
$E_\pm^\mr{in}$-invariant. Here $\rho_+:=\beta$, $\rho_-:=\alpha$ are the two 1-forms on $\gamma$ defined by (\ref{coiso_alpha},\ref{coiso_beta}).
\item \label{C2} For every point $z\in I_\pm^\mr{out}$ the $\infty$-jet of the 1-form $\rho_\pm\in \Omega^1(\gamma)$ at $z$ is
$E_\pm^\mr{out}$-invariant.
\end{enumerate}
The second subspace $C_\mr{out}\subset \Phi_\gamma$ is described similarly where we should interchange superscripts ``in'' and ``out'' in the description of constraints (\ref{C1},\ref{C2}) above.
\item Subspaces $C_\mr{in},C_\mr{out}\subset\Phi_\gamma$ are symplectic w.r.t. symplectic form $\omega$ on $\Phi_\gamma$. Symplectic orthogonals to $C_\mr{in},C_\mr{out}$ in $\Phi_\gamma$ are zero.
\end{enumerate}
\end{proposition}

\begin{proof}
To prove necessity of constraints (\ref{C1},\ref{C2}), assume that a pair $(\phi,\phi_n)\in \Phi_\gamma$ comes from a solution $\tilde\phi$ of the wave equation on an open  neighborhood $V$ of $\gamma$ in $D_\mr{in}$. We can fit into $V$ a topological annulus $D\subset V$ with boundary $\dd D=\gamma'\sqcup\gamma$. The associated involutions $E_\pm(D)$ on $\dd D$ coincide with $E_\pm^\mr{in}(\gamma)$ on some neighborhoods of points $z\in I_\pm^\mr{in}(\gamma)$, which implies constraint (\ref{C1}) by (\ref{rho(L)}). To see (\ref{C2}), fix a sign $\pm$ and fix a point $z\in I_\pm^\mr{out}(\gamma)$. We can choose the annulus $D$ in such a way that the equivalence class of $z$ under equivalence relation $\EE_\pm(D)$ is $\{x,z,y\}$ with $x,y\in\gamma'$. Denote $U'\subset\gamma'$ an open interval on $\gamma'$ bounded by points $x,y$ (among the two possible intervals we choose the $E_\pm(D)$-invariant one). Also fix a neighborhood $U$ of $z$ in $\gamma$; point $z$ splits $U$ into two intervals, $U_{1}$ and $U_2$. Condition (\ref{C2}) on the jet of $\rho_\pm$ at $z$ arises from necessity to smoothly sew an $E_\pm$-invariant 1-form $\rho_\pm$ on $U'$ with $E_\pm^*(\rho_\pm|_{U_1})$ at point $x$ and with $E_\pm^*(\rho_\pm|_{U_2})$ at point $y$.

Conversely, to check sufficiency of (\ref{C1},\ref{C2}), fix $(\phi,\phi_n)\in \CC^\infty(\gamma)^{\times 2}$ satisfying (\ref{C1},\ref{C2}) and fix an annulus $D\subset D_\mr{in}$ with boundary $\dd D=\gamma\sqcup\gamma'$, thin enough, so that for every $z\in I_\pm^\mr{in}(\gamma)$, the neighborhood $U_z$ where we have $E_\pm^\mr{in}$-invariance of $\rho_\pm$ contains the maximal $E_\pm(D)$-invariant neighborhood 
of $z$. Then (\ref{C1}) ensures that 1-form $E_\pm^*(\rho_\pm|_\gamma)$ is smooth on the image of $\gamma$ in $\gamma'$ under $E_\pm$
and (\ref{C2}) ensures that it can be extended to a smooth $E_\pm^\mr{out}(\gamma')$-invariant 1-form on $\gamma'$. Thus we obtain a pair $(\alpha,\beta)\in \Omega^1(\dd D)^{E_-}\times \Omega^1(\dd D)^{E_+}$ which restricts to $(\rho_-,\rho_+)$ on $\gamma$. Then we construct the solution $\tilde\phi$ of wave equation in $D$ as
\be \tilde\phi(\zeta)=\phi(\zeta_0)+\int_{\zeta_0}^\zeta (p_-^*\alpha + p_+^*\beta),\quad \zeta\in D \label{tilde phi}\ee
where $\zeta_0$ is some arbitrary chosen point on $\gamma$, $p_\pm:D\ra \dd D/ E_\pm $ are projections from $D$ to the boundary along light-like lines. Integration path from $\zeta_0$ to $\zeta$ in $D$ is chosen arbitrarily (the integrand is exact since it is closed and restricts to an exact 1-form on one of the two boundary components). By construction, $\tilde\phi$ induces back $(\phi,\phi_n)$ on $\gamma$.

The case of $C_\mr{out}$ is treated similarly.

The calculation of the symplectic orthogonal to $C_\mr{in}$ (case of $C_\mr{out}$ is analogous) in $\Phi_\gamma$ follows the proof of Proposition \ref{prop: omega non-deg}.
We can choose in (\ref{omega as a pairing}) $(\psi,\psi_n)\in C_\mr{in}$ with $\psi_n$ a bump function in neighborhood of any point $z\in\gamma-I$ and vanishing in some open neighborhood of every point of $I$ and $\psi=0$. This proves that $(\phi,\phi_n)\in C_\mr{in}^\perp$ has $\phi(z)=0$. Next, choosing $\psi$ a bump function as above and $\psi_n=0$ we prove that $\phi_n(z)=0$. Thus $C_\mr{in}^\perp=0$. This also implies that $C_\mr{in}$ is symplectic.
\end{proof}

\begin{remark}
Note that $C_\mr{in},C_\mr{out}$ cannot be described in intrinsic terms of $\gamma$, using only the geometric data $(\Gamma,u,\mu)$ as introduced in Section \ref{sec: free boson}: we need more detailed information on the behavior of the metric near $\gamma$ (since we need to know the involutions $E_\pm^{\mr{in},\mr{out}}$ near light-like points of $\gamma$).
\end{remark}

\begin{remark}\label{rem: top annulus C}
Let $D\subset \RR^{1,1}$ be a (topological) annulus bounded by $\dd D=\gamma=\gamma_1\sqcup \gamma_2$ subject to conditions of Section \ref{sec: compact}. We assume that $\gamma_1$ is the inner boundary component and $\gamma_2$ the outer one. Then by Theorem \ref{thm}, the corresponding $L$ is a canonical relation $L\subset C_\mr{out}(\gamma_1)\times C_\mr{in}(\gamma_2)$. Denoting $p_{1,2}$ the projections to the first and second factors in $C_\mr{out}(\gamma_1)\times C_\mr{in}(\gamma_2)$, $p_{1}$ is never injective\footnote{
Indeed, we can take an open subset $U\subset \gamma_2$ such that $E_+(U)\subset \gamma_2$, and a bump function $\psi_U$ on $U$. Then we construct a nonzero element $\psi_U+ E_+^*\psi_U \in L\cap\ker{p_1}$.
}
on $L$ and $p_2$ is never surjective\footnote{
The reason is that $p_2(L)$ is given by $E_\pm$-invariance constraint for 1-forms $\rho_\pm|_{\gamma_2}$ on certain \emph{finite} open subsets  $U_\pm$ of $\gamma_2$, whereas $C_\mr{in}(\gamma_2)$ is given by the constraint (\ref{C1}) of Proposition \ref{prop: C} on \emph{arbitrarily small} neighborhoods of light-like points of $\gamma_2$.
}. Moreover, $p_1$ is surjective and $p_2$ is injective if and only if the following condition holds:
$$I^\mr{out}(\gamma_1)=\varnothing$$
Note that $L$ cannot be a graph of a map $C_\mr{out}(\gamma_1)\ra C_\mr{in}(\gamma_2)$ (nor in the opposite direction).
\end{remark}

\subsection{Conformal invariance}
In case $\dim M=2$, the action (\ref{S}) is invariant under Weyl transformations -- local rescaling of metric $g_{\mu\nu}(x)\ra \Omega(x)\cdot g_{\mu\nu}(x)$ with $\Omega\in \CC^\infty(M)$, $\Omega>0$ . Hence for $F: (M,g)\ra (M',g')$ a conformal diffeomorphism (i.e. $F^*g'=\Omega\cdot g$ with $\Omega> 0$) of 2-dimensional pseudo-Riemannian manifolds, we have $S_{M',g'}(\phi')=S_{M,g}(F^*\phi')$. Thus $F$ induces a symplectomorphism of phase spaces $F^*:\Phi_{\dd M'}\ra \Phi_{\dd M}$ which takes $L_{M'}$ to $L_M$.

In particular, in case of domains $D$ in $\RR^{1,1}$, pairs of a symplectic manifold and an Lagrangian submanifold $(\Phi_{\dd D}, L_D)$ are canonically isomorphic for domains $D$ related by a conformal transformation of $\RR^{1,1}$, e.g. a translation, a Lorentz boost or a rescaling.

Also, Theorem \ref{thm} implies that for $D\subset \RR^2$ a compact domain on the plane endowed with some Lorentzian metric $g_D$, conformally equivalent to a domain $D'\subset \RR^{1,1}$ with Minkowski metric, $L_{D}\subset\Phi_{\dd D}$ is Lagrangian.

\subsection{Hamiltonian for a circle}\label{sec: Hamiltonian for a circle}
Consider 
polar coordinates\footnote{We are using an unconventional radial coordinate, since this choice makes rescaling a translation in $\xi$. In this Section we will sometimes refer to $\xi$ as the ``time'', as the parameter of Hamiltonian dynamics.} $(\xi,\theta)\in \RR\times \RR/(2\pi\ZZ)$ on $\RR^{1,1}$,
$x= e^\xi \cos\theta,\; y= e^\xi\sin\theta$. Phase spaces $\Phi_{S^1_{\xi_0}}$ for circles given by $\xi=\xi_0$ are canonically symplectomorphic for different values of $\xi_0$ by conformal invariance.

For a circle centered at the origin, define a function on the phase space
\be\label{H}
H=\frac{1}{2}\oint_{S^1}d\theta\;\cos(2\theta)\;((\phi_n)^2+(\dd_\theta\phi)^2)\quad \in \CC^\infty(\Phi_{S^1})
\ee
where\footnote{This is a different normalization of the transversal vector field than in Section \ref{sec: compact}. The reason for this choice is that the isomorphism $\Phi_{S^1_{\xi_0}}\simeq \Phi_{S^1_{\xi_1}}$ coming from conformal invariance in these coordinates is just $(\phi,\phi_n)\mapsto (\phi,\phi_n)$.} $\phi_n=\dd_\xi \phi|_{S^1}$.
It generates a Hamiltonian vector field $\check{H}
$
defined by $\iota_{\check{H}}\omega=-\delta H$. Explicitly:
\be
\check{H}=\oint_{S^1} \phi_n\;\frac{\delta}{\delta\phi}+\frac{1}{\cos(2\theta)} (\sin(2\theta)\dd_\theta \phi_n+\dd_\theta(\sin(2\theta) \phi_n+\cos(2\theta)\dd_\theta\phi))\;\frac{\delta}{\delta\phi_n} \label{check H}\ee
Then the infinitesimal evolution in $\xi$
is given by the flow equation for $\check{H}$:
$$\dd_\xi\phi=\check{H}\circ\phi,\quad \dd_\xi\phi_n=\check{H}\circ\phi_n$$
-- this is just an equivalent restatement of the wave equation (\ref{wave eq covariant}) in coordinates $(\xi,\theta)$.

One way to get the function (\ref{H}) is to consider the radial density $\mathcal{L}\in \CC^\infty(\Phi_{S^1})$ of action (\ref{S}) in an annulus $\mr{Ann}_{\xi_0}^{\xi_1}$ defined by $\xi_0\leq\xi\leq \xi_1$,
$$S=\int_{\xi_0}^{\xi_1}d\xi\; \mathcal{L}(\phi|_{S^1_\xi},\dd_\xi\phi|_{S^1_\xi}),\quad
\mathcal{L}=\oint_{S^1}d\theta\;\frac{1}{2}\;\left(\cos(2\theta)\; (\phi_n^2-(\dd_\theta\phi)^2)-\sin(2\theta)\;\phi_n\,\dd_\theta\phi\right)$$
Then one defines
\be H=\oint_{S^1}d\theta\; \phi_n\;\frac{\delta \mathcal{L}}{\delta \phi_n}-\mathcal{L} \label{Legendre}\ee
which yields (\ref{H}). Note that (\ref{Legendre}) is indeed the formula for Legendre transform, but we do not switch to canonical momenta $p=\frac{\delta \mathcal{L}}{\delta \phi_n}=\cos(2\theta)\phi_n-\sin(2\theta)\dd_\theta\phi$.

The Hamiltonian vector field $\check{H}$ (\ref{check H}) is only well-defined on a subspace
$$C_0=\{(\phi,\phi_n)\in \CC^\infty(S^1)^{\times 2}\;|\; (\dd_\theta\phi_n-\dd_\theta\phi)|_{\theta\in\{\pm\frac{\pi}{4},\pm\frac{3\pi}{4}\}}=0\}\quad \subset \Phi_{S^1}$$
due to $\delta H$ not being in the image of the map of vector bundles $\omega^\#:T\Phi_{S^1}\ra T^*\Phi_{S^1}$ (which is injective by weak non-degeneracy of $\omega$, but not an isomorphism) unless one restricts the base to $C_0\subset \Phi_{S^1}$.
More precisely, $\check{H}$ is defined as a section of the pullback of the tangent bundle $T\Phi_{S^1}$ to $C_0$, but it is not generally tangent to $C_0$. However, one may further restrict $\check{H}$ to a smaller subspace $C_1\subset C_0$ where it is tangent to $C_0$; subspace $C_1$ is given by certain restrictions on 3-jets of $(\phi,\phi_n)$ at light-like points on $S^1$. To find the maximal subspace of $\Phi_{S^1}$ on which $\check{H}$ is defined as a tangent vector field, one can iterate this process: cf.\ the Gotay--Nester--Hinds (GNH) geometric constraint algorithm \cite{GNH}, \cite{BPV}. This way one finds a sequence of subspaces
$\Phi_{S^1}\supset C_0\supset C_1\supset C_2\supset\cdots$ where $\check{H}$ on $C_{k+1}$ is tangent to $C_k$, with $C_k$ given by constraints on $(2k+1)$-jets of boundary fields at light-like points of $S^1$. The process does not stabilize at a finite step, and the maximal subspace where $\check{H}$ is defined as a tangent vector field is $C_\infty=\cap_k C_k$ which coincides with $C_\mr{out}(S^1)$ given by constraint (\ref{C2}) of Proposition \ref{prop: C}.

Integrating the vector field $\check{H}$ to a \emph{flow} on $C_\mr{out}(S^1)$ is equivalent to writing the evolution relation $L\subset \bar\Phi_{S^1}\times \Phi_{S^1}$ for the geometric annulus $\mr{Ann}_{0}^{\xi}$ (we are assuming $\xi>0$) as the graph of a map $F_\xi:C_\mr{out}(S^1)\ra C_\mr{out}(S^1)$. This is impossible
due to issues with existence/uniqueness for the 
initial value problem
for the wave equation on the annulus (cf. Remark \ref{rem: top annulus C}).
Specifically, projections $p_{1,2}: \Phi_{S^1}\times \Phi_{S^1}\ra \Phi_{S^1}$ restricted to $L$ yield a diagram
\be C_\mr{out}(S^1)\stackrel{p_1}{\twoheadleftarrow} L \stackrel{p_2}{\hookrightarrow} \underbrace{C_\mr{in}(S^1)}_{\subset\, C_\mr{out}(S^1)} \label{L annulus diagram}\ee
where neither map is an isomorphism.

However, the flow of $\check{H}$ in \emph{negative} time $-\xi<0$ exists as a map
$$F_{-\xi}=p_1(p_2^{-1}(\bt)\cap L):\quad C(-\xi)\ra C_\mr{out}(S^1)$$
where
$C(-\xi)$ is the subspace of $C_\mr{out}(S^1)$ defined as
\begin{multline*}
C(-\xi)=p_2(L)
=\{(\phi,\phi_n)\in \CC^\infty(S^1)^{\times 2}\;|\\
\alpha(\phi,\phi_n)\mbox{ is $E_-$-invariant on }   \underbrace{\left(\frac\pi4-\theta_0, \frac\pi4+\theta_0\right)\cup \left(-\frac{3\pi}{4}-\theta_0, -\frac{3\pi}{4}+\theta_0\right)}_{U_-(\xi)\subset S^1},\\
\beta(\phi,\phi_n)\mbox{ is $E_+$-invariant on }   \underbrace{\left(-\frac{\pi}{4}-\theta_0, -\frac\pi4+\theta_0\right)\cup \left(\frac{3\pi}{4}-\theta_0, \frac{3\pi}{4}+\theta_0\right)}_{U_+(\xi)\subset S^1}
\}
\end{multline*}
where
$p_2$ and $L$ are as in the diagram (\ref{L annulus diagram}),
$\theta_0=\arccos(e^{-\xi})$ and involutions on $S^1$ are $E_-: \theta\leftrightarrow \pi/2-\theta$, $E_+:\theta\leftrightarrow 
-\pi/2
-\theta$; $\alpha$ and $\beta$ are the 1-forms defined by (\ref{coiso_alpha},\ref{coiso_beta}). Note 
that $C(-\xi)\subset \Phi_{S^1}$ is \emph{not} a symplectic subspace; also $F_{-\xi}$ is not injective. What happens instead is that $C(-\xi)\subset \Phi_{S^1}$ is coisotropic, with
\begin{multline} \label{C(-xi) perp}
C(-\xi)^\perp=\ker F_{-\xi}=p_2(\ker p_1\cap L)=\\
=\{(\phi,\phi_n)\in \CC^\infty(S^1)^{\times 2}\mbox{ such that } \alpha|_{S^1-U_-(\xi)}=0,\; \beta|_{S^1-U_+(\xi)}=0,\\
\alpha|_{U_-(\xi)}\mbox{ is $E_-$-invariant},\; \beta|_{U_-(\xi)}\mbox{ is $E_+$-invariant},\;
\phi(\pi/4)+\int_{\pi/4}^{\pi/4+\theta_0}\alpha=0\}
\end{multline}
Formula (\ref{C(-xi) perp}) for $\ker F_{-\xi}$ follows from restricting the solution 
(\ref{tilde phi}) to the inner boundary circle. 
Coincidence of the kernel of $F_{-\xi}$ with the symplectic orthogonal to $C(-\xi)$ follows from Theorem \ref{thm}:
\begin{multline*}C(-\xi)^\perp= \{u=(\phi,\phi_n)\subset \Phi_{S^1}\;|\; \forall s\in L,\; \langle \underbrace{(0,u)}_{\in \bar\Phi_{S^1}\times \Phi_{S^1} },s\rangle_{\bar\Phi_{S^1}\times \Phi_{S^1}}\} =\\
 = p_2(L^\perp\cap\; 0\times \Phi_{S^1}) =p_2(L\cap\; 0\times \Phi_{S^1})=\ker F_{-\xi}
\end{multline*}
Thus $F_{-\xi}$ descends to the symplectic reduction $\underline{C}(-\xi)=C(-\xi)/C(-\xi)^\perp$ and yields an isomorphism 
 of symplectic spaces
$$\underline{F}_{-\xi}:\quad  \underline{C}(-\xi) \stackrel{\sim}{\ra} C_\mr{out}(S^1)$$
which is a symplectomorphism, since before reduction $F_{-\xi}$ pulls back the symplectic structure on $C_\mr{out}(S^1)$ to the presymplectic structure on $C(-\xi)$, as follows from isotropicity of $L$, the graph of $F_{-\xi}$: 
for any pair of elements 
$u,v\in C(-\xi)$ we have
$$
\langle F_{-\xi}(u),  F_{-\xi}(v) \rangle_{\Phi_{S^1}} - \langle u,v \rangle_{\Phi_{S^1}}
=-\langle \underbrace{(F_{-\xi}(u), u)}_{\in L} \, , \, \underbrace{(F_{-\xi}(v), v)}_{\in L}\rangle_{\bar\Phi_{S^1}\times \Phi_{S^1}} =0
$$
With some abuse of terminology, one may call $\underline{F}_{-\xi}$ the ``reduced flow'' 
of the Hamiltonian vector field $\check{H}$ in negative time $-\xi<0$. Then it is reasonable to define the reduced flow in positive time $\xi>0$ to be the inverse map:
$$\underline{F}_{+\xi}= (\underline{F}_{-\xi})^{-1}:\quad C_\mr{out}(S^1) \stackrel{\sim}{\ra}  \underline{C}(-\xi)$$


Note that the reduced Hamiltonian flow does not satisfy the usual semigroup law $\underline{F}_{\xi+\xi'}=\underline{F}_{\xi'}\circ \underline{F}_{\xi}$, since the range of $\uF_\xi$ and the domain of $\uF_{\xi'}$ do not match. Instead we have the following composition law. First consider flows in negative time. The map $\uF_{-\xi}: C(-\xi)/C(-\xi)^\perp\ra C_\Out(S^1)$ can be restricted to a subspace $C(-\xi-\xi')/C(-\xi)^\perp$; this restriction is an isomorphism $C(-\xi-\xi')/C(-\xi)^\perp \stackrel\sim\ra C(-\xi')$. The latter induces
an isomorphism of quotients $\uF_{-\xi,-\xi'}:\;\underline{C}(-\xi-\xi')\stackrel\sim\ra \underline{C}(-\xi')$. Then the composition law is:
\be \uF_{-\xi-\xi'}=\uF_{-\xi'}\circ \uF_{-\xi,-\xi'} \label{composition of flows neg time}\ee
In other words, we take the symplectic reduction of the 3 spaces in the upper row of the diagram
\be \label{composition of flows CD}
\begin{CD}
C(-\xi-\xi')@>F_{-\xi}|_{C(-\xi-\xi')}>> C(-\xi') @>F_{-\xi'}>> C_\Out(S^1) \\
@VVV @VVV \\
C(-\xi)@>F_{-\xi}>> C_\Out(S^1)
\end{CD}
\ee
by $C(-\xi-\xi')^\perp$ (done in two steps: reduction by $C(-\xi)^\perp$ and then by $C(-\xi-\xi')^\perp/C(-\xi)^\perp$),\; $C(-\xi')^\perp$ and $\{0\}$, 
respectively. Vertical arrows in (\ref{composition of flows CD}) are inclusions of subspaces of $\Phi_{S^1}$; composition of the two arrows in the upper row is $F_{-\xi-\xi'}$.

For the composition of reduced flows in positive time, we take the inverse of (\ref{composition of flows neg time}) and interchange $\xi\leftrightarrow\xi'$, obtaining
\be \uF_{\xi+\xi'}=\uF_{\xi',\xi}\circ \uF_{\xi} \ee
where $\uF_{\xi',\xi}=\uF_{-\xi',-\xi}^{-1}:\; \uC(-\xi)\stackrel{\sim}{\ra} \uC(-\xi-\xi')$ is the reduction of the restriction $\uF_{\xi'}|_{C(-\xi)}:\; C(-\xi)\stackrel\sim\ra C(-\xi-\xi')/C(-\xi')^\perp$ by $C(-\xi)^\perp$.

\begin{remark} The Hamiltonian (\ref{H}) descends to the symplectic reduction $\underline{C}(-\xi)$. To see this, note that one can rewrite (\ref{H}) in terms of 1-forms (\ref{coiso_alpha},\ref{coiso_beta}) as
$$H=\oint_{S^1} -\cot\left(\theta-\frac\pi4\right)\;\iota_{\dd_\theta}\alpha\cdot\alpha+ \cot\left(\theta+\frac\pi4\right)\;\iota_{\dd_\theta}\beta\cdot\beta$$
Applying this to a point $u+v\in\Phi_{S^1}$ with $u\in C(-\xi)$ and $v\in C(-\xi)^\perp$ we obtain
\begin{multline*}H(u+v)-H(u)=\\
=\oint_{S^1} -\cot\left(\theta-\frac\pi4\right)\; \iota_{\dd_\theta}(2\alpha_u+\alpha_v)\cdot \alpha_v +
\cot\left(\theta+\frac\pi4\right)\; \iota_{\dd_\theta}(2\beta_u+\beta_v)\cdot \beta_v \\
=\int_{U_-(\xi)}\underbrace{-\cot\left(\theta-\frac\pi4\right)\; \iota_{\dd_\theta}(2\alpha_u+\alpha_v)\cdot \alpha_v}_{E_-\mr{-invariant}} +
\int_{U_+(\xi)}\underbrace{\cot\left(\theta+\frac\pi4\right)\; \iota_{\dd_\theta}(2\beta_u+\beta_v)\cdot \beta_v}_{E_+\mr{-invariant}}\\
=0
\end{multline*}
Thus $H$ does indeed descend to $\uC(-\xi)$. Moreover, the Hamiltonian vector field $\check{H}$ descends to the reduction too. This follows from the explicit formulae for the action of $\check{H}$ on the 1-forms $\alpha,\beta$:
$$\check{H}\alpha=-\dd_\theta\left(\cot\left(\theta-\frac\pi4\right)\cdot \alpha\right),\qquad \check{H}\beta=-\dd_\theta\left(\cot\left(\theta+\frac\pi4\right)\cdot \beta\right)$$
which imply that for $\check{H}$ viewed as a linear map $C_\Out\ra C_\Out$,  both subspaces $C(-\xi)$ and $C(-\xi)^\perp$ are invariant.
\end{remark}


\subsubsection{Banach vs.\ Fr\'echet}\label{ss:BvF}
The impossibility to integrate the vector field $\check{H}$ into a flow on $C_\Out(S^1)$ comes from the fact that since we required from the start that fields are smooth, $\Phi_{S^1}=\CC^\infty(S^1)^{\times 2}$ is naturally equipped with Fr\'echet (but not Banach) topology and hence the
Picard--Lindel\"of theorem for existence and uniqueness of integral trajectories for $\check{H}$ does not apply. We could have chosen a different model for the space of fields from the start, e.g., setting the space of fields to be $F_{D}=\CC^2(D)$ and requiring only $\CC^2$-differentiability for the boundary $\dd D$ in case of a general domain. The phase space then is $\Phi_{\dd D}=\CC^2(\dd D)\times \CC^1(\dd D)\ni (\phi,\phi_n)$, equipped with standard Banach topology. The proof of Theorem \ref{thm} goes through in this setting without any change and, being Lagrangian, $L_D\subset \Phi_{\dd D}$ is automatically closed, and hence a Banach (complete) subspace. In this setting we can try to pass to the Hamiltonian formalism on annuli, with $H$ and $\check{H}$ still given by (\ref{H},\ref{check H}). Then proceeding with the GNH construction as above, we construct a sequence of subspaces $\Phi_{S^1}\supset C_0\supset C_1\supset\cdots$ where $C_k$ becomes a subset of $\CC^{k+3}(S^1)\times \CC^{k+2}(S^1)$ (since an application of $\check{H}$, viewed as a linear map $C_k\ra C_{k-1}$, decreases the regularity by 1 due to the derivatives appearing in (\ref{check H})), with constraints on $(2k+1)$-jets at light-like points of $S^1$, as before. In the end, the maximal subspace $C_\infty$ of $\Phi_{S^1}$, where $\check{H}$ is defined and to which it is tangent, is $C_\infty=\cap_k C_k$. Note that $C_\infty\subset \Phi_{S^1}$ is not a complete subspace (already $C_0$ is not),
hence again the Picard-Lindel\"of theorem does not apply. 
Note also that in the Banach setting $C_\infty\neq C_\Out(S^1)$ since the r.h.s., defined as in Section \ref{sec: C},
has only $\CC^2\times \CC^1$ regularity (with constraints on the 1-jets of the 1-forms $\alpha,\beta$ at light-like points, as opposed to $\infty$-jets arising in the Fr\'echet setting, cf. Proposition \ref{prop: C}, (\ref{C2})).

\subsection{Relational representation of the little 2-disks operad}
Let $\E$ be the operad of little 2-disks \cite{May}, with $\E(n)$ the configuration space of $n$ numbered disjoint (geometric) disks inside a disk of radius $1$ centered at the origin in Euclidean $\RR^2$; these configurations can be viewed as domains $D\subset \RR^2$ obtained by cutting $n$ small disks out of a unit disk. Composition $\circ_i:\E(m)\times \E(n)\ra \E(m+n-1)$ for $1\leq i\leq m$ consists in shrinking 
an element of $\E(n)$ and gluing it into an element of $\E(m)$ instead of the $i$-th disk of the latter.

Part of the data of classical field theory defined by action (\ref{S}) on Minkowski plane is the morphism of operads
\be Z:\E \ra \mr{IsoRel}(\Phi) \label{E2 map}\ee
where $\Phi=\Phi_{S^1}$ is the phase space for the unit circle in $\RR^{1,1}$ (radius and origin are in fact irrelevant due to conformal invariance). For a symplectic space $V$ we denote $\mr{IsoRel}(V)$ the operad of isotropic relations,
\begin{multline*}\mr{IsoRel}(n)=\{\underbrace{V\times\cdots\times V}_n\not\ra V\}=\\=
\{L\subset \underbrace{\bar V\times\cdots\bar V}_n\times V \; 
|\; L\; \mr{an\; isotropic\; subspace}
\}
\end{multline*}
where $\not\ra$ is the symbol for an isotropic relation, bar stands for changing the sign of symplectic form. Composition in $\mr{IsoRel}$ is the set theoretic composition of relations.
Morphism 
$Z$
sends an element of $\E(n)$, viewed as a compact domain $D\subset \RR^{1,1}$ with $n$ ``incoming'' boundary circles and one ``outgoing'' boundary circle, to the corresponding evolution relation $L_D\subset \Phi_{\dd D}\simeq\bar\Phi^{\times n}\times\Phi$, which is canonical (Lagrangian) by Theorem \ref{thm}. The fact that 
$Z$ is indeed a morphism of operads, i.e. is consistent w.r.t. the operadic composition,
is an expression of the general gluing property of classical field theory (here it
simply amounts to the fact that a function $\phi$ on a glued domain $D_1\cup D_2$ solves the wave equation iff its restrictions to $D_{1,2}$ solve the wave equation).

More generally, one can introduce a colored operad $\tilde\E$, with colors being closed curves on $\RR^{1,1}$ modulo conformal transformations and elements of $\tilde\E(n)$ being general compact domains with $n+1$ boundary components, with composition defined (when the colors match) by conformal transformation of one domain and gluing in the hole in another domain. Then we have a morphism of colored operads from $\tilde \E$ to the colored operad of isotropic relations $\Phi_{\gamma_1}\times\cdots\times\Phi_{\gamma_n}\not\ra \Phi_{\gamma_{n+1}}$ with the same set of colors: conformal classes of curves $\gamma_1,\ldots,\gamma_{n+1}$.

Note that this discussion is very general: we only used the general gluing property of field theory, conformal invariance (which is specific for dimension $\n=2$ in case of action (\ref{S})) and the fact that evolution relations are Lagrangian (and in particular isotropic).

\subsection{Free field theories and Lefschetz duality}
An abstract way to view a free classical field theory, natural from the standpoint of the Batalin--Vilkovisky formalism on manifolds with boundary \cite{CMR}, is as follows. One associates to an $\n$-manifold $M$ (possibly endowed with some geometric data, depending on the field theory model in question) a complex of vector spaces $F_M^\bt$ with differential $Q_M$ equipped with a degree $-1$ non-degenerate pairing $\omega_M^{(k)}:F_M^k\otimes F_M^{1-k}\ra\RR$, satisfying $\omega^{(k)}_M(X,Y)=\omega^{(1-k)}_M(Y,X)$ for $X,Y\in F_M$,  and to a closed $(\n-1)$-manifold $\Sigma$ a cochain complex $\Phi^\bt_\Sigma$ with differential $Q_{\dd \Sigma}$, equipped with degree $0$ symplectic structure -- a non-degenerate pairing $\omega^{(k)}_\Sigma:\Phi_\Sigma^k\otimes \Phi_\Sigma^{-k}\ra\RR$ satisfying $\omega_\Sigma^{(k)}(x,y)=-(-1)^k \omega_\Sigma^{(k)}(y,x)$ for $x,y\in\Phi_{\Sigma}$. To the inclusion of the boundary $\Sigma=\dd M\hookrightarrow M$ the field theory associates a chain projection $\pi_M:F^\bt_M\ra \Phi_{\dd M}^\bt$ intertwining the differentials $Q_M$ and $Q_{\dd M}$. The differential $Q_M$, the projection $\pi_M$ and the pairings $\omega_M$, $\omega_{\dd M}$ are required to satisfy the following coherence condition:
\be\omega_M(Q_M X, Y)-(-1)^{\deg X}\omega_M(X,Q_M Y)=\omega_{\dd M}(\pi_M(X),\pi_M(Y)) \label{Lefschetz Q skew symmetry}\ee
for $X,Y\in F_M$.

The short exact sequence
$$\ker\pi_M\hookrightarrow F_M^\bt \stackrel{\pi_M}{\twoheadrightarrow} \Phi_{\dd M}^\bt   $$
induces a long exact sequence in $Q$-cohomology:
\be \cdots\ra H^k_{Q_M}(\ker \pi_M)\ra H^k_{Q_M}\stackrel{\pi_*}{\ra} H^k_{Q_{\dd M}}\stackrel{\beta}{\ra} H^{k+1}_{Q_M}(\ker \pi_M)\ra\cdots \label{LES of rel cohom}\ee
The pairings $\omega_M$, $\omega_{\dd M}$ induce well-defined pairings on cohomology
\begin{eqnarray}
()_M:&& H^k_{Q_M}\otimes H^{1-k}_{Q_M}(\ker \pi)\ra \RR,\label{Lefschetz}\\
(,)_{\dd M}:&& H^k_{Q_{\dd M}}\otimes H^{-k}_{Q_{\dd M}}\ra\RR \label{Poincare}
\end{eqnarray}
In many cases \cite{CMR} these pairings can be proven to be non-degenerate. In particular, for abelian Chern-Simons theory, $(,)_M$ is the Lefschetz duality between de Rham cohomology of a 3-manifold and cohomology relative to the boundary, whereas $(,)_{\dd M}$ is the Poincar\'e duality for de Rham cohomology of the boundary 2-manifold.

The non-degeneracy of the pairing (\ref{Lefschetz}) in the second argument can be shown\footnote{Indeed, one has $\mr{im}(\pi_*)^\perp=\{[x]\in H^\bt_{Q_{\dd M}}\;|\; (\pi_*[Y],[x])_{\dd M}=0\;\forall\; [Y]\in H^\bt_{Q_M}\}$. Using the property $(\pi_*[Y],[x])_{\dd M}=(-1)^{\deg[Y]+1}([Y],\beta[x])_M$ following from (\ref{Lefschetz Q skew symmetry}),
we see that $\mr{im}(\pi_*)^\perp=\beta^{-1}\ker_2(,)_M$, where we denoted $\ker_2(,)_M$ the kernel of the map $H^\bt_{Q_M}(\ker\pi)\ra (H^{1-\bt}_{Q_M})^*$ induced by the pairing (\ref{Lefschetz}). Thus $\mr{im}(\pi_*)^\perp=\ker\beta=\mr{im}(\pi_*)$ if and only if $\ker_2(,)_M$ vanishes.}
 to be equivalent to the property of being Lagrangian for $\mr{im}(\pi_*)\subset H^\bt_{Q_{\dd M}}$.

In the case of the theory defined by the action (\ref{S}), the space of fields $F_M$ is a two-term complex (owing to the absence of gauge symmetry) with $F^0_M=C^\infty(M)\ni \phi$, $F^1_M=\Omega^\n(M)\ni \phi^+$, differential $Q_M:\phi\mapsto d*d \phi$ and pairing $\omega_M(\phi,\phi^+)=\int_M \phi\wedge \phi^+$.
The boundary phase space is a one-term complex $\Phi^0_{\dd M}=\Phi_{\dd M}$ with zero differential and symplectic structure $\omega_{\dd M}$ described in Section \ref{sec: free boson}. The exact sequence (\ref{LES of rel cohom}) becomes in this case
\be \label{LES free boson} 0\ra
\{\phi\;\left|\; \begin{array}{l}d*d\phi=0,\\ \pi_M(\phi)=0\end{array}\right.\}
\ra EL_M\stackrel{\pi_*=\pi_M}{\longrightarrow}\Phi_{\dd M}\stackrel{\beta}{\longrightarrow} \frac{\Omega^\n(M)}{\{d*d\phi\;|\; \pi_M(\phi)=0\}}\ra
\frac{\Omega^\n(M)}{\{d*d\phi\}}
\ra 0\ee
and $\mr{im}(\pi_*)=L_M\subset \Phi_{\dd M}$. 
Thus, whenever Conjecture \ref{Conj} holds for $M$, the ``Lefschetz duality'' (\ref{Lefschetz})
$$EL_M\otimes \frac{\Omega^n(M)}{\{d*d\phi\;|\; \pi_M(\phi)=0\}}\ra \RR$$
is non-degenerate (non-degeneracy in the first term is trivial, whereas for the second term one really needs that $L_M$ is Lagrangian). The pairing between the rightmost and the leftmost terms of (\ref{LES free boson}),
\be \frac{\Omega^\n(M)}{\{d*d\phi\}}\otimes \{\phi\;|\; d*d\phi=0,\; \pi_M(\phi)=0\}\ra \RR \label{Lefschetz bulk zero modes}\ee
is trivially non-degenerate in the second factor, whereas non-degeneracy in the first factor is non-obvious and constitutes a natural extension of Conjecture \ref{Conj}. In the case of Riemannian signature, (\ref{Lefschetz bulk zero modes}) becomes, by the Hodge--Morrey decomposition
theorem \cite{CDGM}, the pairing
$$H^\n(M)\otimes H^0(M,\dd M)\ra \RR$$
which is a special case of the standard Lefschetz duality and is indeed non-degenerate.
On the other hand, for $M$ a compact domain in the Minkowski plane as in Theorem \ref{thm}, one can easily show that both outmost terms of (\ref{LES free boson}) vanish.

\subsection{More general Lorentzian surfaces}

By inspection of its proof, Theorem \ref{thm}  generalizes straightforwardly to the case of a compact surface $M$ with smooth boundary endowed with a Lorentzian metric $g$ smooth up to the boundary, if the following conditions hold:
\begin{enumerate}[(a)]
\item \label{Lorentz a)}The two null-distributions
$\dd_+\subset TM$, $\dd_-\subset TM$ of the metric $g$ induce, as in Section \ref{sec: Evol relation and involutions}, two piecewise smooth involutions $E_\pm$ on the boundary $\dd M$ with finitely many points removed.
\item \label{Lorentz b)} For each choice of the sign $\pm$, the restriction map $C^\infty(M)^{\dd_\pm}\ra C^\infty(\dd M)^{E_\pm}$ is surjective. Here $C^\infty(M)^{\dd_\pm}$ stands for the space of smooth functions on $M$, constant along the distribution $\dd_+$ or $\dd_-$, respectively.
\item \label{Lorentz c)} The first Betti number of the cohomology of $M$ relative to the boundary vanishes, $\dim H^1(M,\dd M)=0$.
\end{enumerate}

\begin{remark}
\begin{enumerate}[i.]
\item Obviously, conditions (\ref{Lorentz a)}, \ref{Lorentz b)}, \ref{Lorentz c)}) hold if $(M,g)$ is conformally equivalent to a domain $D\subset \RR^{1,1}$ in the Minkowski plane satisfying conditions (\ref{assump A}, \ref{assump B}, \ref{assump C}) of Section \ref{sec: compact}.
\item For $M$ a domain $D\subset\RR^{1,1}$ in the Minkowski plane, condition (\ref{Lorentz b)}) is equivalent to assumption (\ref{assump C}) of Section \ref{sec: compact}, i.e.\ the assumption that lightlike points of the boundary are neither inflection, nor undulation points.
\item The presence of a family of null-curves originating at $\dd M$ and asymptotically approaching a limiting closed null-curve in $M$ spoils both conditions (\ref{Lorentz a)}) and (\ref{Lorentz b)}), see the example in Section \ref{sec: Misner}.
\item For $(M,g)$ a general Lorentzian surface, if $M'\subset M$ is a sufficiently small disk cut out of $M$, conditions (\ref{Lorentz a)},\ref{Lorentz b)},\ref{Lorentz c)}) hold for $(M',g|_{M'})$ and thus the corresponding evolution relation $L_{M'}\subset\Phi_{\dd M'}$ is Lagrangian.
\end{enumerate}
\end{remark}

\subsection{An example where $L$ is not Lagrangian: the Misner space}\label{ss-Misner}
\label{sec: Misner}
Consider the following Lorentzian manifold (the Misner space \cite{Mis}): $M=S^1\times [-1,1]$ --- a cylinder with coordinates $x\in \RR/2\pi\ZZ$, $y\in [-1,1]$ --- endowed with the  Lorentzian metric
$$g=dx\,dy-y\, dx^2$$
The corresponding null-distributions on $M$ are:
$$\dd_+=\dd_y,\quad  \dd_-= -\dd_x - y\, \dd_y$$
In particular, the ``in-boundary'' $S^1\times \{-1\}$ is spacelike and the ``out-boundary'' $S^1\times \{1\}$ is timelike. Moreover, the circle $S^1\times \{0\}$ is a leaf of the distribution $\dd_-$, i.e. a closed null-curve.

The equations for the integral curves of distributions $\dd_\pm$ (the null-curves) are
$$\frac{dx}{dy}=\frac1y,\quad \frac{dx}{dy}=0$$
for the $\dd_-$- and $\dd_+$-curves, respectively. In particular, all $\dd_-$-curves originating at either boundary circle asymptotically approach the null-cycle $S^1\times \{0\}$. On the other hand, the $\dd_+$-curves are simply the vertical lines $\{x\}\times [-1,1]$, for any $\{x\}\in S^1$.

The phase space associated to the boundary of $M$ by the construction of Section \ref{sec: free boson} is
$$\Phi_{\dd M}=\underbrace{C^\infty(S^1)\times C^\infty(S^1)}_{\Phi_{\dd_\In M}}\times \underbrace{C^\infty(S^1)\times C^\infty(S^1)}_{\Phi_{\dd_\Out M}}\ni (\phi^\In,\phi_n^\In,\phi^\Out,\phi_n^\Out)$$
where we have chosen the transversal vector field to be $n=2\dd_y-\dd_x$ at the in-boundary and $n=2\dd_y+\dd_x$ at the out-boundary. The symplectic form (\ref{presymp form}) on the phase space is
$$\omega=\oint_{S^1}dx\, (\delta \phi^\In\wedge \delta \phi^\In_n + \delta \phi^\Out\wedge \delta \phi^\Out_n) $$

For the evolution relation, consider first the ``global'' Euler-Lagrange space (in the sense of Section \ref{sec: Evol rel non-simply connected}):
$$EL^\glob=\{\phi=F+G\;\in C^\infty(M)\;|\;F,G\in C^\infty(M),\; \dd_- F=\dd_+ G=0\}$$
Since all $\dd_-$-curves asymptotically approach the single closed null-curve $S^1\times \{0\}$, the function $F$ is forced by continuity to be constant (which can be absorbed into $G$). Thus the restriction to the ``global part'' of the evolution relation is
\begin{multline*}
\label{Misner Lglob}
L^\glob=\pi(EL^\glob)=\\
=\{(\phi^\In=g(x),\;\phi_n^\In=-\dd_x g(x),\;\phi^\Out=g(x),\;\phi_n^\Out=\dd_x g(x))\quad \in \Phi_{\dd M}\;|\; g\in C^\infty(S^1)\}
\end{multline*}

The symplectic orthogonal to $L^\glob$ is readily calculated to be
\begin{multline*}
(L^\glob)^\perp=\{(\phi^\In,\;\phi^\In_n,\; \phi^\Out,\; \phi^\Out_n)\quad\in \Phi_{\dd M}\;| \\ |\; \dd_x\phi^\In(x)-\phi^\In_n(x)-\dd_x \phi^\Out(x)-\phi_n^\Out(x)=0\;\forall x\in\RR/2\pi\ZZ\}
\end{multline*}
which implies that $L^\glob$ is isotropic and
$$\dim\, (L^\glob)^\perp/L^\glob=\infty$$
(since in $(L^\glob)^\perp$ one can choose $\phi^\In,\phi_n^\In,\phi^\Out$ as independent functions, whereas in $L^\glob$ they are all expressed in terms of a single function $g$).

The true Euler--Lagrange space, where the possible multivaluedness of $F,G$ is taken into account, is given by (\ref{EL non simply connected}). In the case of the Misner geometry, $\iota_{\dd_-}\kappa=0$ implies that $\int_{S^1\times\{0\}}\kappa=0$, hence $\kappa$ defines zero cohomology class in $H^1(M)$ and therefore $\kappa$ (and hence $\lambda$ too) is exact. This implies that there is no distinction between $EL$ and $EL^\glob$ in the case at hand. Thus $L=L^\glob$ and, by the discussion above, \textit{the evolution relation $L$ is isotropic, but not Lagrangian.}

It is easy to check that also the two halves of the Misner cylinder considered above, $M_1=S^1\times [-1,0]$ and $M_2=S^1\times [0,1]$, produce non-Lagrangian evolution relations.

\end{document}